\documentclass[journal]{IEEEtran}

\usepackage{amsmath,url}
\usepackage{amssymb}
\usepackage{amsthm}
\usepackage{graphicx}
\usepackage[outdir=./]{epstopdf}
\usepackage{tikz,multirow}
\usepackage{verbatim}
\usepackage{bm}
\usepackage{subcaption,float,cite}
\usepackage{caption}

\newtheorem{lemma}{Lemma}{}
  \newtheorem{thm}{Theorem}
  \newtheorem{theorem}{Theorem}
  
  \newtheorem{cor}[thm]{Corollary}
  
  \newtheorem{definition}[thm]{Definition}

\def\R{\mathbb{R}}
\def\P{\mathbb{P}} 
\def\ie{{\em i.e.}}

\def\i{\mathbf{1}} 

\def\N{\mathcal{N}}

\def\ds{\sigma}

\def\ni{\mathcal{N}_i}
\def\nj{\mathcal{N}_j}

\def\lf{\beta}
\def\I{\mathcal{I}}

\def\x{\textbf{x}}

\def\p{\hat{p}}
\def\b{\hat{b}}
\def\H{\hat{H}}

\def\bv{\textbf{b}_v}
\def\bn{\textbf{b}_f}
\def\pv{\textbf{p}_v}
\def\pn{\textbf{p}_f}

\def\z{{\bf{z}}}

\def\lt{\tilde{\lambda}}

\def\be{\bm{\beta}}

\begin{document}

\title{Adaptive CSMA under the SINR Model: Efficient Approximation Algorithms for Throughput and Utility Maximization}

\author{Peruru Subrahmanya Swamy, Radha Krishna Ganti,~\IEEEmembership{Member,~IEEE,} Krishna Jagannathan,~\IEEEmembership{Member,~IEEE}
\thanks{A part of this work \cite{allerton_bethe} has been presented at IEEE Annual
  Allerton Conference on Communication, Control, and Computing (Allerton) 2015, held at Monticello, IL, USA}
\thanks{P. S. Swamy, R. K. Ganti and K. Jagannathan are with the Department of Electrical Engineering, IIT Madras, Chennai, India 600036. Email:$\{$p.swamy, rganti, krishnaj$\}$@ee.iitm.ac.in}}

\maketitle

\begin{abstract}
We consider a Carrier Sense Multiple Access (CSMA) based scheduling algorithm for a single-hop wireless network under a realistic Signal-to-interference-plus-noise ratio (SINR) model for the interference. We propose two local optimization based approximation algorithms to efficiently estimate certain attempt rate parameters of CSMA called fugacities. It is known that adaptive CSMA can achieve throughput optimality by sampling feasible schedules from a Gibbs distribution, with appropriate fugacities. Unfortunately, obtaining these optimal fugacities is an NP-hard problem. Further, the existing adaptive CSMA algorithms use a stochastic gradient descent based method, which usually entails an impractically slow  (exponential in the size of the network) convergence to the optimal fugacities. To address this issue, we first propose an algorithm to estimate the fugacities, that can support a given set of desired service rates. The convergence rate and the complexity of this algorithm are independent of the network size, and depend only on the neighborhood size of a link. Further, we show that the proposed algorithm corresponds exactly to performing the well-known Bethe approximation to the underlying Gibbs distribution. Then, we propose another local algorithm to estimate the optimal fugacities under a utility maximization framework, and characterize its accuracy. Numerical results indicate that the proposed methods have a good degree of accuracy, and achieve extremely fast convergence to near-optimal fugacities, and often outperform the convergence rate of the stochastic gradient descent by a few orders of magnitude.

\end{abstract}


\section{Introduction}

The problem of link scheduling for maximum throughput has been widely studied, particularly with emphasis on the maximum-weight scheduling framework, developed in \cite{tassiulas,TE93}. In spite of its throughput maximizing property, maximum-weight scheduling requires centralized control, and necessitates the solution of an NP-hard problem for each scheduling decision. Several works have attempted to modify the maximum-weight algorithm, so as to make it more amenable to simple, distributed implementation \cite{chaporkar2008throughput,dimakis2006sufficient,wu06}. However, these greedy algorithms do not achieve full throughput.

In a series of recent papers \cite{libin,dshah,qcsma}, a class of distributed link scheduling algorithms called adaptive CSMA algorithms have been proposed and proven to be throughput optimal, \ie, they can support any achievable service rate vector. The central idea behind these algorithms lies in using a reversible Markov chain to sample feasible schedules from a product form distribution called the Gibbs distribution \cite[Chapter 7]{bremaud}. Specifically, each link adaptively adjusts its transmission attempt rate (also known as its fugacity) in order to ensure sufficient average service rate. 

In order to support a given feasible service rate vector using CSMA, the corresponding fugacities have to be computed. Unfortunately, determining the appropriate fugacities corresponding to the desired service rates is an NP-hard problem \cite{libin}. In \cite{libin}, the optimal fugacities are computed as a solution to an optimization problem (here after referred to as the Gibbsian problem), using a stochastic gradient descent algorithm. Each iteration of the gradient descent requires estimating the average service rates under the current iterate of the fugacities, which in turn entails waiting for the underlying Markov chain to reach steady-state. This `mixing time' of the underlying Markov chain could be very large (exponential in the size of the network), depending on the network load and topology \cite{hardness}, \cite{fast_mixing}. Therefore, the existing adaptive CSMA algorithms \cite{libin} do not provide a practical way to estimate the fugacities, although they effectively support the desired service rates once the optimal fugacities are estimated \cite{parallel_chains}. The main focus of this paper is in proposing efficient methods to estimate the fugacities, under a realistic SINR (signal-to-interference-plus-noise ratio) model for the interference. Specifically, we consider the following two scenarios under which fugacities are to be computed:
\begin{itemize}
\item The service rate requirements of the links are known. The objective is to support these average service rates.
\item The service rate requirements are not known, but each link has a utility function of its average service rate. The objective is to maximize the sum utility of the network.
\end{itemize}

 A simple conflict graph based interference model ~\cite{libin, dshah, qcsma}  is widely used in the wireless context due to its simplicity and tractability, although it does not adequately capture the complex nature of the wireless interference \cite{graph_limitations}. More specifically, a {conflict graph} based interference model ignores the fact that whether or not two links can transmit simultaneously, depends on the transmission state of the other links and their spatial locations. Some recent papers extend the adaptive CSMA framework to a more realistic interference models like the SINR model \cite{sinr_mimo, sinr_mimo_journal}, and Rayleigh faded channels \cite{ncc_paper}. However, these papers also essentially employ stochastic gradient descent on a Gibbsian function to estimate the fugacities, and hence suffer from impractically slow convergence rates. 
 
For the conflict graph based interference model assumed in \cite{libin}, approximate but efficient methods to compute the fugacities have been proposed \cite{bp_csma,bethe_jshin} using a popular variational technique called the Bethe approximation \cite{yedidia}. However, the solutions given in \cite{bp_csma,bethe_jshin}, cannot be directly extended to SINR based interference model. This is because the conflict graph based interference model corresponds to a simple pair-wise interaction model \cite {yedidia}, while the SINR model involves higher order interactions. The presence of these higher order interactions makes the extension non-trivial. Even for graphical models with higher order interactions, there are well known algorithms like Belief propagation (BP) \cite{yedidia} for solving the Bethe approximation problem. However, in the context of adaptive CSMA under the SINR model, they can directly be used only to estimate the service rates given the fugacities, but not the other way around. 




We start with the Gibbsian optimization problem corresponding to the optimal fugacities, and approximate this global problem by decoupling it into local optimization problems at each link. The local problems are identical in structure to the global problem, and are referred to as the local Gibbisan problems. The dimension of the local problem at a link is equal to the size of its immediate neighbourhood and hence typically small, and independent of the network size. Therefore these local Gibbsian problems can be efficiently solved in a scalable fashion. The local solutions are then suitably combined to obtain an approximate solution to the global problem.

We prove that the solution of our local Gibbs optimization method corresponds exactly to the celebrated Bethe approximation \cite{yedidia} to the global Gibbsian optimization problem. The accuracy of the Bethe approximation has been empirically evidenced in various fields \cite{bethe_emprical}. Therefore, it is reasonable to expect a fair degree of accuracy in the context of CSMA as well. In fact, numerical results indicate that in order to obtain the level of accuracy in the fugacities obtained by using our proposed method, the stochastic gradient descent method \cite{libin} takes an inordinately long time, often running into tens of millions of time units even for fairly small networks. Therefore, in practical terms, our algorithm will operate with substantially smaller convergence time, compared to the original implementation of adaptive CSMA.

It is worth noting that our algorithm is robust to gradual changes in the desired service rates, as well as to changes in the network topology, since the local solutions can be efficiently re-computed for the new set of requirements. On the other hand, the stochastic gradient descent is likely to take a very long time to converge to its new operating point. 

The impractically slow mixing time of the CSMA Markov chain is also known to result in poor delay performance \cite{hardness, fast_mixing}. Recent works like~\cite{parallel_chains, parallel_chains2} have shown improved  delay performance by employing several parallel instances of this Markov chain. However, the result in \cite{parallel_chains} assumes that the required optimal fugacities can be pre-computed and given to their algorithm. Our local algorithms which efficiently estimate these fugacities can be used in conjunction with the techniques in \cite{parallel_chains, parallel_chains2} to obtain a practical CSMA algorithm with good throughput and delay properties.

%
%
%
%
%
%
%

The remainder of this paper is organized as follows. In Section \ref{sec_model}, we describe the SINR based interference model, and review the adaptive CSMA algorithm. In Section \ref{local_gibbs}, we introduce the local Gibbsian problems and propose our algorithm for computing the fugacities for a given service rate requirements. Section \ref{preliminaries_bethe} provides a brief review of the Bethe approximation, as relevant to adaptive CSMA. In Section \ref{gibbs_bethe}, we derive our main result which establishes the equivalence between the local Gibbs optimization, and the Bethe approximation. In Section \ref{special_cases}, we consider the conflict graph model as a special case of the SINR model and derive closed-form expressions for the local Gibbsian problems. In Section \ref{util_max}, we propose a local algorithm to solve the utility maximization problem and quantify its performance. In Section \ref{simulations}, we present numerical results to confirm the fast convergence, and Section \ref{conc} concludes the paper.

\section{Model and Preliminaries} \label{sec_model}
We consider a single-hop wireless network and model the links using a bipole model, introduced in \cite{baccelli}. In a bipole model, each transmitter is associated with a receiver on the Euclidean plane.  A transmitter and its corresponding receiver are referred to as a link. Let $\N$ denote the set of all the links in the network. Let $|\N| =N$ be the total number of links. Let $r_{ii}$ denote the distance between the transmitter and receiver of link $i$. For simplicity, we assume\footnote{The results in this paper do not require this assumption. This is just to keep the expressions concise.} that a link distance $r_{ii}$ is much smaller than the distances of the transmitter and the receiver from the other links.  With this assumption, we can think of links as points in the Euclidean space $\R^2$. Let $r_{ji}$ denote the distance between the links $i, j$. We assume a time slotted model.

\emph{Interference model:}
We consider the standard path-loss model $\|d\|^{-\alpha}, \alpha>2$, where $d$ is the distance between a receiver and a transmitter, and $\alpha$ is the path loss exponent. Let $P_i$ denote the transmit power of link $i$. We assume white Gaussian thermal noise  at all the receivers with variance $w$. Let $\x(t)=[x_i(t)]_{i=1}^N$ denote the schedule of the network at time $t$. Specifically, $x_i(t)=1$ denotes that the link $i$ is active (transmitting) in time slot $t$. If there is no ambiguity, we will also use $\x$ to denote $\x(t)$.

  Although all the active links in the network can potentially contribute to the interference, the aggregate interference from the transmitters beyond a certain distance can be safely neglected \cite{radius_approx1}, \cite{radius_approx2}. This approximation is standard in the literature \cite{sinr_mimo} and this distance, referred to as the \emph{close-in} radius is denoted by $R_I$. Let $\N_i:= \lbrace k \; | \; r_{ik} \leq R_I \rbrace $. For convenience, let link $i$  be also included in the set $\N_i$. We refer to the links in $\N_i \setminus \{i\}$ as the neighbors of link $i$.  The neighborhood relationship can be represented by an \emph{interference graph} $G(V,E)$. $V$ is the set of links in the network and two links share an edge if they are within a distance $R_I$. Then the total interference power at link $i$ is given by
\begin{align}
I_i(\x) = \sum\limits_{ \{ j \in \N_i \;|\; j \neq i,\; x_j=1\}} P_j r_{ji}^{-\alpha}. \label{eq_int1}
\end{align}
Then, the SINR at link $i$ is given by
\begin{align}
\gamma_{i}(\x) =\frac{P_i r_{ii}^{-\alpha}}{I_i(\x)+w}. \label{eq_sinr}
\end{align}
 
\emph{Reception model:}
We assume that, in each time slot, a single packet of data is transmitted from each active transmitter. If the received SINR at the corresponding receiver exceeds a pre-determined threshold T, \ie, $\gamma_i(\x) \geq T$, the packet is successfully received.

\emph{Rate region:}
A schedule $\x \in \{0,1\}^{\N}$ is said to be \emph{feasible}, if all the active links in the schedule meet the required SINR constraint, \ie , $\gamma_i(\x) \geq T, \; \forall i \text{ such that } x_i =1.$ The set of all the feasible schedules is denoted by $\mathcal{I}$. In our scenario, since each link transmits one data packet whenever it is successful, the long-term service rate of a link is equal to the fraction of time the link is successful. The \emph{rate region} $\Lambda$, which is defined as the set of all the possible service rates is given by the convex hull of the feasible schedules in $\mathcal{I}$.  Hence, $\Lambda= \{ \sum_{ \x \in \mathcal{I}} \alpha_{\x} \x \; | \; \sum_{\x \in \mathcal{I}} \alpha_{\x}=1, \; \alpha_{\x} \geq 0, \; \x \in \mathcal{I} \}$. If a link scheduling policy can support any rate vector in the rate region, then the scheduling policy is said to be \emph{rate-optimal}. 

\emph{Adaptive CSMA:} 
\label{subsec_acsma}
We briefly review the adaptive CSMA algorithm \cite{libin, sinr_mimo_journal}. In this algorithm, each link $i$ is associated with a fugacity $\lambda_i > 0$ which defines the underlying Gibbs distribution. In each time slot, a randomly selected link $i$ is allowed to update its schedule $x_i(t)$ based on the information in the previous slot:
\begin{itemize}
\item If its SINR is inadequate, \ie,  $\gamma_i(\x(t-1)) \leq T$, then $x_i(t)=0$.
\item If $\gamma_i(\x(t-1)) \geq T$, then link $i$ exchanges control messages with its neighbors, to find if they can meet their SINR requirements if link $i$ gets activated. If any of its neighbors cannot meet its requirement, then $x_i(t)=0$.
\item If all the neighbors can meet their SINR requirements even if link $i$ gets activated, then $x_i(t)=1$ with probability $\frac{\lambda_i}{1+\lambda_i}$, and $x_i(t)=0$ with probability $\frac{1}{1+\lambda_i}$.
\end{itemize}
\emph{Remarks on Implementation:}  A challenge in the second step of the above algorithm is to ensure that a newly scheduled link does not alter the SINR requirements of its transmitting neighbours. This problem is typically addressed by introducing a control subslot during which nodes exchange control packets to determine the feasibility of transmissions. More specifically, \cite{sinr_mimo_journal} proposed a control protocol which includes a three-way handshake of control packets: Ready-To-Send (RTS), Clear-To-Send (CTS), and REJECT. In this protocol, if a link is selected to update its status, its transmitter sends an RTS in the control subslot. If any active neighbouring link in the previous schedule $\x(t-1)$ fails to meet the required SINR, the neighbour broadcasts a REJECT signal and disapproves the transmission request. If no REJECT signal is broadcasted, the selected link proceeds to transmit. Detailed descriptions of this handshake protocol can found in \cite[Section III-C]{sinr_mimo_journal}.

It can be shown \cite[Proposition 1]{qcsma} that the adaptive CSMA algorithm induces a Markov chain on the state space of the schedules $\{0,1\}^\N$. Further, the stationary distribution of the Markov chain, parametrized by the fugacity vector $\lambda=[\lambda_i]_{i=1}^N$, is given by:
\begin{align}
p(\x)&= \frac{1}{Z}\prod\limits_{  j \; : \; x_j=1}\lambda_j \; \i (\x \in \mathcal{I}) , & \forall \x \in \{0,1\}^N,  \label{eq_dist}
\end{align}
where $\i(\x \in \mathcal{I})$ is an indicator of $\x$ being a feasible schedule, and $Z$ is the normalizing constant. Then, due to the ergodicity of the Markov chain, the long-term service rate of a link $i$ denoted by $s_i$ is equal to the marginal probability that link $i$ is active, \ie, $p_i(x_i=1)$. Thus, the service rates and the fugacities are related as follows:
\begin{align}
s_i&= p_i(1) = \sum\limits_{\x \; : \; x_i=1 } Z^{-1} \prod\limits_{ j \; : \; x_j=1}\lambda_j, \; \; \forall i \in \N, \label{eq_serv_fug}
\end{align}
where $p_i(1)$ denotes $p_i(x_i=1)$.
The adaptive CSMA algorithm can support any service rate in the rate region provided appropriate fugacities are used for the underlying Gibbs distribution \cite[Theorem 5]{fast_mixing} . 

If the desired service rates are known, these fugacities can be obtained by solving the system of equations in \eqref{eq_serv_fug}.   In \cite[Section 3.3]{libin_book}, it is shown that solving this system of equations can be posed as a convex optimization problem. The main idea is explained as follows. Given a set of service rates from the rate region, by the definition of rate region, there should exist some distribution $\alpha(\x)$ on the state space of feasible schedules $\I$ which would support the required service rates. Note that the definition of rate region does not impose that the distribution $\alpha(\x)$ is a Gibbs distribution. In \cite[Section 3.3]{libin_book}, the authors find a Gibbs distribution that is close to this distribution $\alpha(\x)$. It is achieved by minimizing the KL divergence~\cite[Section 3.3]{libin_book} between $\alpha(\x)$ and the family of Gibbs distributions parameterized by the fugacity vectors. Further it has been shown \cite [Section 3.3]{libin_book} that the resulting minimization problem is equivalent to the following Gibbsian optimization problem referred to as the global problem.  

\noindent \textbf{The global Gibbsian problem:}
\begin{align}
\ln \lambda &= \arg\max_{r \in \R^{\N}} G(r), \label{opt_global} \\
\text{where  }G(r)&:=\sum\limits_{k \in \N} s_k r_{k} - \ln \Big(\sum\limits_{y \in \I} \exp\Big(\sum\limits_{k \in \N} y_k r_{k} \Big)\Big). \nonumber
\end{align} 
Here $\{s_i\}_{i \in \N} \in \Lambda$ are the desired service rates.

\noindent \emph{Remark:}  To understand that the global Gibbsian problem \eqref{opt_global} solves the system of equations in \eqref{eq_serv_fug}, we can simply set  $\frac{\partial G(r)}{\partial r_i}=0$ to obtain
\begin{align*}
s_i - \frac{\sum_{y \in \I : y_i=1} \exp\Big(\sum_{k=1}^N y_k r^*_k\Big)}{\sum_{y \in \I} \exp\Big(\sum_{k=1}^N y_k r^*_k\Big)}=0, \; \forall i \in \N.
\end{align*}
Observe that for $\lambda_i= e^{r_i^*}$, the above equations essentially boil down to the desired equations in \eqref{eq_serv_fug}.

A distributed stochastic gradient descent algorithm was proposed in \cite{libin} to solve \eqref{opt_global}. However, estimating the gradient of $G(r)$ in a distributed manner entails the underlying Markov chain of the CSMA algorithm to converge to steady-state, which takes an impractically long time in general \cite{parallel_chains}. 

In the next section, we consider this scenario where the links know their target service rates. We provide an efficient and scalable approximation to this problem, by proposing the following local Gibbsian problems. The solutions of these local problems are appropriately combined to estimate the solution to the global problem.

\section{The local Gibbsian problems} \label{local_gibbs}
We now introduce some definitions required for the description of the local Gibbsian problems.

\emph{Local schedule:} Let $\x^{(j)} \in \{0,1\}^{\N_j}$, be the set of variables corresponding to the transmission status of the link $j$ and its neighbors, \ie,
$\x^{(j)}:=\left\lbrace x_k \; | \; k \in \N_j  \right\rbrace.$
We refer to $\x^{(j)}$ as the \emph{local schedule} at $j$. Further, from \eqref{eq_int1}, it can be observed that the SINR of a link depends only on the local schedule. Hence the SINR at a link $j$ can be viewed as a function of $\x^{(j)}$, \ie, $\gamma_j(\x) = \gamma_j(\x^{(j)})$. Also, recall that a schedule $\x$ is said to be feasible, if all the active links in the schedule meet their required SINR threshold. Thus, $\i(\x \in \mathcal{I})$ can be factorized over the local schedule variables $\{\x^{(j)}\}_{j=1}^N$ as $\i(\x \in \mathcal{I})= \prod\limits_{j \;: \;x_j=1} \i (\gamma_j(\x^{(j)}) \geq T).$ Then, \eqref{eq_dist} can be written as
\begin{align}
p(\x)&= \frac{1}{Z} \prod\limits_{j \;: \;x_j=1} \lambda_j \i (\gamma_j(\x^{(j)}) \geq T), & \forall \x \in \{0,1\}^N. \label{eq_succint}
\end{align}

\emph{Local feasiblity:} A local schedule $y = [y_k]_{k \in \N_j} \in \left\lbrace 0, 1 \right\rbrace^{\N_j}$ at link $j$ is said to be \emph{feasible}, if either the link $j$ is inactive (\ie, $y_j=0$), or it is active and meets the required threshold SINR \ie, $(y_j=1$ and $\i (\gamma_j(y) \geq T))$. The set of all the feasible local schedules at $j$ is denoted by $\I_j$. It can be observed from \eqref{eq_succint} that $p(\x)$ assigns zero probability to a schedule $\x$ if any of its local schedule is infeasible.

\emph{Local service rate vector:} Let $\{s_i\}_{i=1}^N$ be the set of service rates of all the links in the network. Then the local service rate vector at link $j$ denoted by $s^{(j)}$ be defined as the set of all the service rates corresponding to link $j$ and its neighbors, \ie, $s^{(j)}:=\left\lbrace s_k \; | \; k \in \N_j  \right\rbrace.$

\emph{Local capacity region:} The local capacity region at link $j$ is defined as the convex hull of the local feasible schedules at link $j$ given by
\begin{align}
\Lambda_j= \{ \sum_{ z \in \mathcal{I}_j} \alpha_{z} z \; | \; \sum_{z\in \mathcal{I}_j} \alpha_{z}=1, \; \alpha_{z} \geq 0, \; z \in \mathcal{I}_j \}. \label{eq_loc_cap}
\end{align}
Now, we define the local Gibbsian problem at link $j$ as follows:
\begin{align}
\lf_{j} = \arg\max\limits_{r \in \R^{\N_j}} F(r), \label{opt_local_alternate0}
\end{align} 
where $\lf_j:=[\lf_{jk}]_{k \in \N_j}$, and the function $F: \R^{\N_j} \rightarrow \R$ is
\begin{align*}
F(r):=\sum\limits_{k \in \N_j} s_k r_{k} - \ln \Big(\sum\limits_{y \in \I_j} \exp\Big(\sum\limits_{k \in \N_j} y_k r_{k} \Big)\Big).
\end{align*} 
Observe that, the local problems are structurally similar to the global problem, except that $\I$ in the global problem is replaced by $\I_j$, and $\N$ is replaced by $\N_j$. In particular, the dimension of the local problem at link $i$ is just $|\N_i|$. The solutions to these local problems are referred to as the \emph{local fugacities}. In particular, at each link $j \in \N$, there is a local fugacity vector  $\lf_j:=[\lf_{jk}]_{k \in \N_j}$.

\emph{Local algorithm:}
Here, we propose a simple and distributed algorithm (Algorithm 1) to solve the local Gibbsian problems and subsequently compute the approximate global fugacities by combining the local solutions using \eqref{eq_lf_gf0}. These approximate global fugacities can be directly used in the CSMA algorithm instead of adapting the fugacties using a stochastic gradient descent on the global problem which usually doesn't converge in practical time scales. Each link in the network executes the following algorithm in parallel.
\noindent\rule[0.5ex]{\linewidth}{0.5pt}
\textbf{Algorithm 1:} Local Gibbsian method at link $j$\\
\noindent\rule[0.5ex]{\linewidth}{0.5pt}
\; \; \emph{Input:} ($s_{k}$,  $k \in \nj$); \hspace{0.25cm}
\emph{Output:} $\lt_j$.
\begin{enumerate}
\item Obtain the service rates ($s_{k}$,  $k \in \nj$) from the neighbours.
\item Compute the local fugacities ($\lf_{jk}$,  $k \in \nj$) by solving the local problem \eqref{opt_local_alternate0}  using the Newton's method.
\item From each neighbour $k \in \N_j$, obtain the local fugacity $\lf_{kj}$.
\item Compute the approximate global fugacity $\lt_j$ as
\begin{align}
\lt_j&= \left(\frac{1-s_j}{s_j}\right)^{|\N_j|-1}\prod\limits_{k \in \N_j} e^{\lf_{kj}}. \label{eq_lf_gf0}
\end{align}
\end{enumerate}
\noindent\rule[0.5ex]{\linewidth}{0.5pt}

\emph{Information exchange:} The algorithm requires only two steps of information exchange with the neighbours. Once in the first step, to obtain the service requirements of the neighbours, and again in the third step to obtain the local fugacities computed at the neighbours. Except for these two information exchanges, the algorithm is fully distributed and can be executed independently at each link.

\emph{Computational complexity:}
The implementation of the Newton's method \cite[Section 9.5]{boyd} in the second step of our algorithm is feasible. This is because the gradient and the Hessian of the local objective function $F(r)$ can be analytically computed since the dimension of the problem is small. The exact expressions for the gradient and Hessian at link $j$ can be computed using certain marginals of the distribution  $\b_j(\x^{(j)}) =  {Z_j^{-1}} \exp (\sum_{k \in \N_j} x_k r_k ),  \forall \x^{(j)} \in \I_j$, where $Z_j$ is a normalization constant. 

For $k \in \N_j$, let $m_k(r)$ represent the probability $\P(x_k=1)$ under the distribution $\b_j$. Similarly, for $i,k \in \N_j$, let $m_{ik}(r)$ denote the probability $\P(x_i=1, x_k=1)$ under the same distribution $\b_j$. 
Then, the gradient and Hessian of the function $F(r)$ are given by
\begin{align*}
[\nabla(F(r))]_k&= s_k-m_k(r), \; \;  k \in \N_j, \\
[\nabla^2(F(r))]_{ik}&= \begin{cases}
m_{ik}(r)-m_i(r)m_k(r) , \; \;  i,k \in \N_j, i \neq k\\
m_k(r)-m_k(r)^2, \; \; i=k.
\end{cases}
\end{align*}
The computation of the gradient and the Hessian requires the information about the local feasible schedules at a link. Specifically, this information is required to compute the normalization constant $Z_j$.  These computations are feasible because the $O(2^{|\N_j|})$ complexity involved in this computation scales only with the size of the local neighborhood, and is independent of the total size of the network which could be substantially large. In particular, in spatial networks where the neighborhood size does not scale with the network size, our algorithm is order optimal.
%
%
%
%
%

In Section \ref{gibbs_bethe}, we prove that the approximate global fugacites $\{\lt_j\}_{j=1}^N$ obtained using the local Gibbsian method \eqref{eq_lf_gf0} correspond exactly to performing the well known Bethe approximation to the global Gibbsian problem. In the next section, we review the Bethe approximation technique.

\section{Review of the Bethe approximation} \label{preliminaries_bethe}

 We now introduce some terminology required to describe the Bethe approximation technique \cite{yedidia}. 
\subsection{Product form distribution}
Let $S$ be a finite set, and let $X_i, i=1, 2, \dots, N$, be random variables each taking values in $S$. The joint PMF (probability mass function) of the random variables is succinctly denoted as $p(\x)$, where $\x=\{x_1,x_2, \dots, x_N \}$. Suppose that $p(\x)$ factors into a product of $M$ functions and is given by,
\begin{align}
p(\x)&= \frac{1}{Z} \prod\limits_{j=1}^M f_j(\x^{(j)}), \hspace{0.7cm} \x \in S^N.  \label{eq_factors}
\end{align}
The function $f_j(\cdot)$ has arguments $\x^{(j)}$ that are some non-empty subset of $\x=\{x_1,x_2, \dots, x_N \}$. Here $Z$ is a normalization constant. Further, the product form distributions are generally represented using a graph called factor graph \cite{yedidia} that has two sets of nodes namely variable nodes and factor nodes corresponding to the variables and the functions in \eqref{eq_factors} respectively. An edge is drawn between a variable node and a factor node if the variable is an argument of that factor function.

For the product form distribution considered in \eqref{eq_factors}, we are interested in certain marginal probabilities called the variable marginals and the factor marginals.

\noindent \emph{Variable node marginals: }
The marginal probability distribution $p_i({x_i})$ corresponding to a variable node $i$, is obtained by summing $p(\x)$ over the variables corresponding to all other variable nodes, \ie, $p_i({x_i})= \sum_{\x \setminus x_i} p(\x)$, $x_i \in S$, $ i=1,\dots, N.$

\noindent \emph{Factor node marginals:} Corresponding to each factor node, $j=1, \dots, M$, the marginal probability function $\p_j({\x^{(j)}})$,  is obtained by summing $p(\x)$ over all the variables in $\x \setminus \x^{(j)}$. Let $N_j$ denote the number of arguments in $f_j$, \ie, $N_j= |\x^{(j)}|$. Then, 
$
\p_{j}({\x^{(j)}}) = \sum_{\x \setminus \x^{(j)}} p(\x)$, $\x^{(j)} \in S^{N_j}.$

Let $\pv:=\{p_i\}_{i=1}^N$, $\pn:=\{\p_j\}_{j=1}^M$ denote the collection of all the variable node marginals and  the factor node marginals corresponding to the distribution $p(\x)$ respectively.

The computation of these marginal probability distributions requires the computation of the normalization constant $Z$, which is an NP-hard problem \cite{yedidia}. Next, we discuss the notions of Gibbs free energy (GFE) and the Bethe free energy (BFE) which provide a variational characterization of the normalization constant and the marginal distributions.

\subsection{Gibbs free energy}
Consider a probability distribution $p(\x)$ of the form \eqref{eq_factors} for which we are interested in finding the normalization constant $Z$, and the marginals. Let $b(\x)$ be some distribution on $S^N$. Here, $p(\x)$ is referred to as the \emph{true} distribution, and $b(\x)$ is referred to as the \emph{trial} distribution as it will be used to estimate the true distribution $p(\x)$.
We now introduce the notion of energy function which is required to define the Gibbs free energy \cite{yedidia}. For a given distribution $p(\x)$ in the form of \eqref{eq_factors}, the energy function $E(\x) : S^N \rightarrow \R$ is defined as $E(\x)= - \sum_{j=1}^{M} \ln f_j(\x^{(j)}).$ Observe that the energy function $E(\x)$ completely specifies the distribution $p(\x)$ in \eqref{eq_factors}, as $p(\x)= \frac{1}{Z} e^{-E(\x)}.$ Using the definition of the energy function, the GFE can be defined as follows. 
\begin{definition}
Consider an energy function $E(\x)$ which corresponds to a true distribution $p(\x)$. Let $b(\x)$ be a trial distribution. Then the Gibbs free energy $F_G(b)$ is defined as $F_G(b)= U_G(b) - H_G(b)$, where the term $U_G(b)=\sum_{\x \in S^N} b(\x) E(\x)$ is called the average energy, and $H_G(b)= - \sum_{\x \in S^N} b(\x) \ln b(\x)$ is the entropy of the distribution $b(\x)$.
\end{definition}
The GFE provides a variational characterization of the normalization constant $Z$ of $p(\x)$ as $- \ln Z = \min_b F_G(b).$ Further, if the minimum of $F_G(b)$ is achieved at $b^*$, then $p(\x)=b^*(\x), \; \forall \x \in S^N.$ 

Here, the optimization is over all the possible distributions on $S^N$. However, as $N$ becomes large, this procedure is intractable, as the optimization variables take exponentially large memory to store.  Moreover, this method only computes the partition function, but doesn't explicitly compute the marginals. We hence take recourse to the Bethe approximation, which is an approximate, but a more practical technique to estimate the marginals explicitly. In particular, the Bethe approximation does two approximations: First, the GFE $F_G(b)$ is approximated using the BFE $F_B(b)$ defined in \eqref{eq_bfe_eg}. Second, the optimization of the BFE is performed over a restricted set of distributions, which will be described in \eqref{eq_bethe_min}.

\subsection{Bethe approximation}
\subsubsection{Bethe free energy (BFE)}
We first define the Bethe entropy of a distribution $b(\x)$, which is a function of the factor and variable marginals of $b(\x)$. The Bethe entropy is given by
\begin{align*}
H_B\left( \bn, \bv \right)
&=\sum\limits_{j=1}^M \H_j(\b_j) - \sum\limits_{i=1}^N (d_i-1) H_i(b_i).
\end{align*}
Here, $d_i$ is the degree of the variable node $i$ in the factor graph, and $\H_j(\b_j)= -\sum_{y \in S^{N_j}} \b_j(y) \ln \b_j(y)$, $
H_i(b_i)= -\sum_{y \in S} b_i(y) \ln b_i(y)$ denote the entropies of the factor marginal $\b_j(\x^{(j)})$, and the variable marginal $b_i(x_i)$ respectively.
\begin{definition}
Consider a true distribution $p(\x)$ of the form \eqref{eq_factors}. Let $b(\x)$ be a trial distribution with factor and variable marginals given by $\bn=\{\b_j\}_{j=1}^M$ and $\bv=\{b_i\}_{i=1}^N$. Then, the Bethe free energy $F_B(b)$ corresponding to the true distribution $p(x)$ is defined as
\begin{align}
F_B(b)&=F_B\left( \bn, \bv \right)=U_B(\bn, \bv) - H_B(\bn, \bv),  \label{eq_bfe_eg}
\end{align}
\vspace{-5mm}
\begin{align}
\text{where } U_B(\bn, \bv)=-\sum\limits_{j=1}^M \sum\limits_{{\x^{(j)}} \in S^{N_j}}  \b_j(\x^{(j)}) \ln f_j(\x^{(j)}). \label{eq_ub}
\end{align}
\end{definition}
 \subsubsection{Bethe optimization (BO)}
The following optimization problem is referred to as the Bethe optimization.
\begin{align}
&\underset{\bn, \bv} {\arg \min}\; \;  F_B(\bn, \bv), \text{  subject to}  \label{eq_bethe_min}\\
 b_i(x_i) &\geq 0, \hspace{0.4cm} i = 1, \dots, N, \; \;  x_i \in S,\nonumber\\ 
  \sum_{x_i} b_{i} (x_i)&=1, \hspace{0.4cm} i = 1, \dots, N, \nonumber\\ 
\b_j(\x^{(j)}) &\geq 0, \hspace{0.4cm} j = 1, \dots, M, \; \;   \x^{(j)} \in S^{N_j},\nonumber\\
 \sum_{\x^{(j)}} \b_{j} (\x^{(j)})&=1, \hspace{0.4cm} j = 1, \dots, M, \nonumber\\
 \sum_{\x^{(j)} \setminus \{x_i \}} \b_j(\x^{(j)})&= b_i(x_i), \hspace{0.2cm} j=1, \dots, M, \; i \in \N_j,\; x_i \in S ,\nonumber
\end{align}
where $\N_j$ is the set of indices of all the variable nodes associated with the factor node $j$.

\emph{Note:}
A feasible collection $(\bn,\bv)$ of the Bethe optimization problem does not necessarily represent the marginals of any coherent joint distribution over $S^N$, and are therefore referred to as \emph{pseudo-marginals} \cite{yedidia}. In spite of performing the optimization over a relaxed set, the solution of the Bethe optimization $(\bn^*, \bv^*)$ exactly coincides with the marginals of the true distribution $(\pn^*, \pv^*)$, if the underlying factor graph is a tree. Further, for factor graphs with loops, the solution leads to very good estimates of the true marginals \cite{yedidia, bethe_emprical}. In many applications, solving for the global minimizer of \eqref{eq_bethe_min} is not feasible because of its non-convexity. Hence, a stationary point of the BFE satisfying the constraints in \eqref{eq_bethe_min} is considered as a good estimate of the true marginals \cite{yedidia, bethe_emprical}.

\subsection{Bethe approximation for CSMA}
We observe that the stationary distribution of the CSMA Markov chain $p(\x)$ given by \eqref{eq_succint} is in product form \eqref{eq_factors} with the factor functions $\{f_j \}_{ j=1}^N$ defined as
\begin{align}
f_j(\x^{(j)})= \begin{cases}
\lambda_j, & \text{if }  x_j=1, \gamma_j(\x^{(j)}) \geq T, \label{eq_local_factors} \\
  0, & \text{if } x_j=1, \gamma_j(\x^{(j)}) < T, \\
 1,  & \text{if } x_j = 0.
\end{cases}
\end{align}
where $\x^{(j)}$ is the local schedule of the link $j$ defined earlier.

By the definition of local feasibility, $p(\x)$ assigns zero probability to a schedule $\x$, if it is locally infeasible at any link. Hence, the factor node marginals $\pn = \{\p_j\}_{j=1}^N$ have to satisfy $\p_j(y)=0, \; \; \; \forall y \notin \I_j, \;  \forall j \in \N.$

\textbf{BFE for CSMA:}
We shall consider the product form representation of the CSMA stationary distribution $p(\x)$ in \eqref{eq_succint} and compute the BFE. As defined in \eqref{eq_bfe_eg}, the BFE for a given trial distribution $b(\x)$ has two terms $U_B(\bn,\bv)$, $H_B(\bn, \bv)$. Now we compute the first term $U_B(\bn, \bv)$. Specifically, we show that the term $U_B(\bn, \bv)$ which is in general a function of the factor marginals $\bn$, reduces to a function of just the variable marginals $\bv$, subject to the feasibility conditions stated in the following lemma.
\begin{lemma} \label{lemma_u_bethe}
Let $(\bn,\bv)$ be a set of feasible factor and variable marginals of the Bethe optimization problem \eqref{eq_bethe_min}.  If the factor marginals of a trial distribution $\bn$, assign zero probabilities to all the infeasible local schedules, \ie, 
 \begin{align}
 \b_j(y)=0, \;  \forall y \notin \I_j, \;  \forall j \in \N, \label{eq_bn_feasible}
 \end{align}
then $U_B(\bn, \bv)$ reduces to $U_B(\bn, \bv)=$ $U_B(\bv)= -\sum_{i=1}^N b_i(1) \ln\lambda_i.$ Moreover, if $\bn$ does not satisfy \eqref{eq_bn_feasible}, $U_B(\bn, \bv)=  \infty$, which in turn implies $F_B(\bn, \bv)=\infty$.
\end{lemma}
\begin{proof} 
Recalling the definition of local feasibility (introduced in Section \ref{local_gibbs}), we obtain
\begin{align*}
\I_j=\{\x^{(j)}\;|\; x_j = 0\} \cup  \{\x^{(j)}\;|\; x_j=1, \gamma_j(\x^{(j)}) \geq T \}.
\end{align*}
Consider the inner summation of $U_B(\bn,\bv)$ in \eqref{eq_ub}, and expand it into two terms as follows:
\begin{align}
&\sum\limits_{\x^{(j)} \in \{0,1\}^{\N_j}} \b_j(\x^{(j)})  \ln f_j(\x^{(j)}) \nonumber\\
&=\sum\limits_{ \{\x^{(j)} \in \I_j \} } \b_j(\x^{(j)}) \ln f_j(\x^{(j)}) +\sum\limits_{\{\x^{(j)} \notin \I_j \}} \b_j(\x^{(j)}) \ln f_j(\x^{(j)}). \label{eq_summation}
\end{align}
Now, we evaluate the first term of \eqref{eq_summation}.

\vspace{-4mm}
\begin{scriptsize}
\begin{align}
&\sum\limits_{ \{\x^{(j)} \in \I_j \} } \b_j(\x^{(j)}) \ln f_j(\x^{(j)}) \overset{(a)}{=} \nonumber \\
&\sum\limits_{\{\x^{(j)}\;|\; x_j = 0\}} \b_j(\x^{(j)}) \ln f_j(\x^{(j)}) +\sum\limits_{\{\x^{(j)}\;|\; x_j=1, \gamma_j(\x^{(j)}) \geq T \}} \b_j(\x^{(j)}) \ln f_j(\x^{(j)})   \nonumber, \\
&\overset{(b)}{=}\sum\limits_{\{\x^{(j)}\;|\; x_j = 0\}} \b_j(\x^{(j)}) \ln 1 + \sum\limits_{\{\x^{(j)}\;|\; x_j=1, \gamma_j(\x^{(j)}) \geq T \}} \b_j(\x^{(j)}) \ln \lambda_j  \nonumber,  \\
&=\ln \lambda_j \sum\limits_{\{\x^{(j)}\;|\; x_j=1, \gamma_j(\x^{(j)}) \geq T \}} \b_j(\x^{(j)}).\label{eq_lemma_ub1}
\end{align}
\end{scriptsize}
\vspace{-4mm}

where $(a)$ follows from the definition of local feasibility, $(b)$ follows from the defintion of factors in \eqref{eq_local_factors}.

We now show that the second term of \eqref{eq_summation} evaluates to zero. If $\x^{(j)} \notin \I_j$, then from \eqref{eq_local_factors} and the assumption \eqref{eq_bn_feasible}, we have $f_j(\x^{(j)})=0$ and $\b_j(\x^{(j)})=0$ respectively. Hence the summation corresponding to ${\{\x^{(j)} \notin \I_j \}}$ in \eqref{eq_summation} is equal to zero.

Now, we use the fact that $(\bn,\bv)$ is feasible for the Bethe optimization problem \eqref{eq_bethe_min}. Specifically, we use the last equality constraint in \eqref{eq_bethe_min} for $i=j$ case (recall from the definition of $\N_j$, that the link $j$ is also included in the set $\N_j$) to obtain

\vspace{-3mm}
\begin{scriptsize}
\begin{align}
&b_j(1)=\sum\limits_{\{\x^{(j)}\;|\; x_j=1\}} \b_j(\x^{(j)}), \nonumber \\
&=\sum\limits_{\{\x^{(j)}\;|\; x_j=1, \gamma_j(\x^{(j)}) \geq T \}} \b_j(\x^{(j)}) + \sum\limits_{\{\x^{(j)}\;|\; x_j=1, \gamma_j(\x^{(j)}) < T \}} \b_j(\x^{(j)}), \nonumber \\
&\overset{\eqref{eq_bn_feasible}}{=}\sum\limits_{\{\x^{(j)}\;|\; x_j=1, \gamma_j(\x^{(j)}) \geq T \}} \b_j(\x^{(j)}) + 0. \label{eq_lemma_ub2}
\end{align}
\end{scriptsize}
\vspace{-3mm}

Substituting  \eqref{eq_lemma_ub1}, \eqref{eq_lemma_ub2} in \eqref{eq_summation}, we obtain
\begin{align*}
\sum\limits_{\x^{(j)} \in \{0,1\}^{\N_j}} \b_j(\x^{(j)})  \ln f_j(\x^{(j)}) = b_j(1) \ln \lambda_j .
\end{align*}
Using the above equation in the definition of $U_B(\bn,\bv)$ in \eqref{eq_ub} completes the proof of this Lemma.
\end{proof}
Using Lemma \ref{lemma_u_bethe} and \eqref{eq_bfe_eg}, the BFE for any $(\bn, \bv)$ that satisfies the feasibility constraints in \eqref{eq_bethe_min}, \eqref{eq_bn_feasible} is given by
\begin{align}
F_B(\bn,\bv)&=F_B\left( \{\b_j\}_{j=1}^N,  \{b_i\}_{i=1}^N \right),  \label{eq_bfe}\\
&= \sum\limits_{i=1}^N  -b_i(1) \ln\lambda_i  - \H_i(\b_i)  +  (d_i-1) H_i(b_i), \nonumber
\end{align}
where $d_i$ is degree of variable node $i$ in the factor graph. We refer $F_B(\bn,\bv)$ \eqref{eq_bfe} as the BFE corresponding to the set of fugacities $\{\lambda_i\}_{i=1}^N$. It essentially means that $F_B(\bn,\bv)$ is the BFE corresponding to $p(\x)$ in \eqref{eq_succint}. 

\emph{Remark:}  Since we are only interested in the Bethe optimization problem \eqref{eq_bethe_min}, throughout this paper, we implicitly assume that the factor and variable marginals $(\bn,\bv)$ satisfy the feasibility conditions in \eqref{eq_bethe_min}, \eqref{eq_bn_feasible}, and use the BFE $F_B(\bn,\bv)$ expression obtained in \eqref{eq_bfe}.



\section{Equivalence of Local Gibbsian method and the Bethe approximation} \label{gibbs_bethe}
We first prove certain important structural properties of the BFE, which we then use to establish an equivalence between the local Gibbsian method (Algorithm 1) and the Bethe approximation. In particular, the above equivalence hinges on establishing the following two important properties, which are formalized subsequently in Lemmas \ref{lemma_max_entropy} through \ref{lemma_fug}.
\begin{itemize}
\item At a stationary point of the BFE, the relation between the factor and the variable marginals is captured by the local Gibbsian problem.
\item The factor and variable marginals corresponding to a stationary point of the BFE uniquely determine the fugacities.
\end{itemize}
\noindent \emph{Remark:} The above properties are derived for the Bethe approximation under the SINR model considered in this paper. The Bethe approximation under general settings may not satisfy these properties.
\subsection{Characterization of the stationary points of the BFE}
In this subsection, we characterize the stationary points of the Bethe free energy $F_B(\bn,\bv)$ \eqref{eq_bfe} in Lemmas \ref{lemma_max_entropy} and \ref{lemma_fug}.
The following Lemma asserts that the stationary points of $F_B(\bn,\bv)$ satisfy a maximum entropy property.
\begin{lemma} \label{lemma_max_entropy}
 Let $F_B(\bn,\bv)$ \eqref{eq_bfe} denote the BFE corresponding to the fugacities $\{\lambda_i\}_{i=1}^N$. Let $(\bn^*, \bv^*)$ be a stationary point of $F_B(\bn,\bv)$ that satisfies the feasibility constraints in \eqref{eq_bethe_min}, \eqref{eq_bn_feasible}. Then, for each $i\in \N$, the factor marginal $\b_i^*$ is related to its corresponding variable marginals $\{b_j^*\}_{j \in \N_i}$, through the following constrained entropy maximization problem, parametrized by the variable marginals $\{b_j^*\}_{j \in \N_i}$:
 \begin{align}
&\b_i^*=\arg \max\limits_{\b_i} \H_i(\b_i), \text{     subject to} \label{opt_local_entropy}\\
&\b_i(\x^{(i)}) \geq 0, \; \forall \x^{(i)} \in \I_i; \sum_{\x^{(i)} \in \I_i} \b_i(\x^{(i)})=1,\nonumber\\
 &\sum\limits_{\x^{(i)} \setminus \{x_j \}} \b_i(\x^{(i)})= b_j^*(x_j), \hspace{0.2cm}  \; \forall j \in \N_i,\; x_j =1. \label{eq_dual}
\end{align}
 \end{lemma}
\begin{proof}
Proof is provided in Appendix~\ref{proof_max_entropy}.
\end{proof}
Further, it can be shown \cite[Section 3.5]{libin_book} that there is a unique solution to the maximum entropy problem \eqref{opt_local_entropy}. Specifically, the variable marginals $\{b_j^*\}_{j \in \N_i}$ uniquely characterize the corresponding factor marginal $\b_i^*$, through the local Gibbsian problem \eqref{opt_local_alternate}; this is formalized in Lemma \ref{lemma_libin}.

 \begin{lemma} \label{lemma_libin}
Consider the following local Gibbsian problem defined by the variable marginals $\{b_j^*\}_{j \in \N_i}$:
\begin{align}
v_i = \arg\max\limits_{r \in \R^{\N_i}} \sum\limits_{k \in \N_i} b_k^*(1) r_{k} - \ln \Big(\sum\limits_{y \in \I_i} \exp\Big(\sum\limits_{k \in \N_i} y_k r_{k} \Big)\Big). \label{opt_local_alternate}
\end{align} 
Then the corresponding factor marginal (optimal solution of \eqref{opt_local_entropy}) is given by
\begin{align}
\b_i^*(\x^{(i)})& = \frac{1}{Z_i} \exp\Big(\sum\limits_{k \in \N_i} x_k v_{ik} \Big), & \forall \x^{(i)} \in \I_i, \label{eq_opt_factor_marg}
\end{align}
where $v_i=[v_{ik}]_{k \in \N_i}$ is the solution of \eqref{opt_local_alternate}, $Z_i$ is a normalization constant.
 \end{lemma}
\begin{proof}
Proof follows by considering the dual problem of \eqref{opt_local_entropy}. The proof can be found in  \cite[Section 3.5]{libin_book}.
\end{proof}

{Next, in Lemma \ref{lemma_fug}, we prove that the factor and variable marginals at a stationary point of the BFE uniquely characterize the fugacities.}
\begin{lemma} \label{lemma_fug}
Let $F_B(\bn,\bv)$ \eqref{eq_bfe} denote the BFE corresponding to the fugacities $\{\lambda_i\}_{i=1}^N$. Let $(\bn^*, \bv^*)$ be a stationary point of the $F_B(\bn,\bv)$. Then the fugacities are uniquely determined by the factor and variables marginals $(\bn^*, \bv^*)$ as follows:
\begin{align}
\lambda_i&= \left(\frac{1-b^*_i(1)}{b_i^*(1)}\right)^{d_i-1}\prod\limits_{j \in \ni} e^{v_{ji}}, \; \; \forall i \in \N, \label{eq_lf_gf}
\end{align}
where $v_{ji}$ is an element of the vector $v_j=[v_{jk}]_{k \in \N_j}$, that characterizes the factor marginal $\b_j^*$ of some $j\in \N_i$.
\end{lemma}
\begin{proof}
Proof is provided in Appendix \ref{proof_fugacities}.
\end{proof}
In Lemmas \ref{lemma_max_entropy}, \ref{lemma_fug} we have derived the necessary conditions that a stationary point of the BFE should satisfy. In Appendix \ref{proof_complete}, we show that the conditions in these two Lemmas together constitute a sufficient condition for a stationary point of the BFE.
\subsection{Local Gibbsian method gives the Bethe approximated fugacities}
We are now in a position to state a key result of this paper, which asserts that the approximated global fugacities \eqref{eq_lf_gf0} obtained by solving the local Gibbsian problems \eqref{opt_local_alternate0} correspond exactly to the Bethe approximated fugacities. We formalize this result through Definition \ref{def_bethe_fug}, Theorem \ref{thm_bethe_main}.
 
\begin{definition} (Bethe approximated fugacities) \label{def_bethe_fug}
For a given set of service rates $\{s_i\}_{i=1}^N$, a set of fugacities $\{\lambda_i\}_{i=1}^N$ are said to be Bethe approximated fugacities if the following holds:
``Consider the BFE $F_B(\bn,\bv)$ corresponding to the fugacities $\{\lambda_i\}_{i=1}^N$. Then there should exist a stationary point $(\bn^*,\bv^*)$ point of the BFE such that the corresponding variable marginals are equal to the given service rates, \ie, $b_i^*(1)=s_i, \forall i$".
\end{definition}
\begin{theorem} \label{thm_bethe_main} 
Let $\{s_i\}_{i=1}^N$ be the desired service rates.
\begin{itemize} 
\item[1.]  Then the global fugacities $\{\lt_i\}_{i=1}^N$ estimated in \eqref{eq_lf_gf0} are the Bethe approximated fugacities for $\{s_i\}_{i=1}^N$. 
\item[2.] For a given set of service rates, the Bethe approximated fugacities are unique. 
\end{itemize}
\end{theorem}

\begin{proof}
Consider the BFE \eqref{eq_bfe_new} corresponding to the fugacities $\{\lt_i\}_{i=1}^N$. We show there exists a stationary point $(\{\underline{\b}_j\}, \{\underline{b}_j\})$ of the BFE, given by \eqref{eq_thm1}-\eqref{eq_thm2}, such that the corresponding variable marginals $\{\underline{b}_j(1)\}$ \eqref{eq_thm1} are equal to the desired service rates $\{s_i\}_{i=1}^N$.
\begin{align}
F_B(\bn,\bv)&= \sum\limits_{i=1}^N   -b_i(1) \ln\lt_i  - \H_i(\b_i)  +  (d_i-1) H_i(b_i), \label{eq_bfe_new} \\
\underline{b}_j(1)&=s_j, \; \underline{b}_j(0)=1-s_j, \; \; \forall j, \label{eq_thm1}\\
 \underline{\b}_j(\x^{(j)})& = \frac{1}{Z_j} \exp\Big(\sum\limits_{k \in \N_j} x_k \lf_{jk} \Big), \; \; \forall \x^{(j)} \in \I_j, \; \forall j, \label{eq_thm2}
\end{align}
where $\lf_j=[\lf_{jk}]_{k \in \N_j}$ is the local fugacity vector at link $j$ obtained from \eqref{opt_local_alternate0}, and $Z_j$ is the corresponding normalization constant. 

Now we use the structural properties of the BFE derived in Lemmas 2 through 4 to complete the proof. In particular, recall from \eqref{opt_local_alternate0} that the local fugacity vector $\lf_j=[\lf_{jk}]_{k \in \N_j}$ is obtained by solving the local Gibbsian problem \eqref{opt_local_alternate0}, which is same as \eqref{opt_local_alternate}. Hence, due to Lemma \ref{lemma_libin}, it is clear that the set of marginals $(\{\underline{\b}_j\}, \{\underline{b}_j\})$ in \eqref{eq_thm1}-\eqref{eq_thm2} satisfy the maximum entropy property stated in Lemma \ref{lemma_max_entropy}. 

Next, due to the definition \eqref{eq_lf_gf0} of the global fugacities $\{\lt_i\}_{i=1}^N$, the variable and factor marginals defined in \eqref{eq_thm1}-\eqref{eq_thm2} immediately satisfy the condition \eqref{eq_lf_gf} in Lemma \ref{lemma_fug}. 

Therefore, $(\{\underline{\b}_j\}, \{\underline{b}_j\})$ given in \eqref{eq_thm1}-\eqref{eq_thm2} is a stationary point of the BFE \eqref{eq_bfe_new}. Further, note that the corresponding variable marginals $\{\underline{b}_j(1)\}$ \eqref{eq_thm1} at this stationary point are equal to the desired service rates $\{s_i\}_{i=1}^N$. Hence, the fugacities $\{\lt_i\}_{i=1}^N$ are the Bethe approximated fugacities  for $\{s_i\}_{i=1}^N$.

Next, we prove the uniqueness of the Bethe approximated fugacities. For the given service rates $\{s_i\}_{i=1}^N$, let us assume that two sets of fugacities $\{\lambda_i^1\}_{i=1}^N$, $\{\lambda_i^2\}_{i=1}^N$ satisfy the definition of the Bethe approximated fugacities. Let $F_B^1(\bn,\bv)$, $F_B^2(\bn,\bv)$ be the BFE functions corresponding to the  fugacities $\{\lambda_i^1\}_{i=1}^N$, $\{\lambda_i^2\}_{i=1}^N$ respectively. Then there should exist stationary points $(\{\b_j^1\},\{b_j^1\} )$, $(\{\b_j^2\},\{b_j^2\} )$ of their corresponding BFE functions such that the variables marginals are equal to the service rates, \ie, $\{b_j^1(1)\}=\{b_j^2(1)\}=\{s_j\}$. Since the variable marginals are same, due to Lemmas \ref{lemma_max_entropy}, \ref{lemma_libin}, the factor marginals corresponding to these stationary points should be the same. In other words,  \ie,  $(\{\b_j^1\},\{b_j^1\} )= (\{\b_j^2\},\{b_j^2\} )$. Then \eqref{eq_lf_gf} from Lemma \ref{lemma_fug} asserts that $\{\lambda_i^1\}_{i=1}^N = \{\lambda_i^2\}_{i=1}^N$. Hence, for a given set of service rates, the Bethe approximated fugacities are unique.
\end{proof}

\section{Special case - Conflict graph model} \label{special_cases}
Conflict graph model \cite{bethe_jshin} is a special case of the SINR interference model. In the conflict graph model, two links cannot transmit simultaneously, if one link is within the interference range of the other link. Under the conflict graph model, we derive simple closed form expressions for the local fugacities  \eqref{opt_local_alternate0}.
\begin{theorem} \label{thm_lf_cg}
 For the conflict graph model, the local fugacities \ie, solution of the local Gibbsian problem \eqref{opt_local_alternate0} at a link $i$ is given by
\begin{align}
e^{\lf_{ij}}&= \begin{cases} \label{eq_lf_cg}
s_i  \left(1-s_i\right)^{|\N_i|-2} \prod\limits_{k \in \N_i \setminus \{i\}}(1-s_i-s_k)^{-1},\; \text{if } j=i, \\
s_j(1-s_i-s_j)^{-1}, \; \; \; \; \; \text{if } j \in \N_i \setminus \{i\}.
\end{cases}
\end{align}
\end{theorem}
\begin{proof}
Proof is provided in Appendix \ref{proof_lf_cg}. 
\end{proof}
\noindent \emph{Comparison with the results in \cite{bethe_jshin}:}\\
The fugacities for the conflict graph model have been derived in \cite{bethe_jshin}. Here, we derive the global fugacities using our approach and compare to the results in \cite{bethe_jshin}. 
\begin{cor}
 \label{cor_cg}
 For the conflict graph model, the global fugacity at a link $i$ is
\begin{align*}
\tilde{\lambda}_i&= \frac{s_i (1-s_i)^{2|\N_i|-3}}{\prod\limits_{k \in \N_i \setminus \{i\}} (1-s_i-s_k)^2}.
\end{align*}
\end{cor}
\begin{proof}
We obtain this result by substituting the local fugacities obtained in \eqref{eq_lf_cg} in the expression proposed for global fugacities \eqref{eq_lf_gf0}.
\end{proof}
The expression for the global fugacities proposed in \cite{bethe_jshin} is
\begin{align*}
\tilde{\lambda}_i&= \frac{s_i (1-s_i)^{|\N_i|-2}}{\prod\limits_{k \in \N_i \setminus \{i\}} (1-s_i-s_k)},
\end{align*}
which is different from the expression derived in Corollary \ref{cor_cg}. In essence, these expressions are different because, the factorizations considered in the two cases are different. Specifically, it is well known that for a given product form distribution, the Bethe approximation technique could lead to different results, corresponding to different factorizations of the product form distribution \cite[Chapter 2, Page 10]{book_martin}.

%

To elucidate this point, we consider the conflict graph interference model, and present two natural factorizations that lead to the same CSMA distribution  \eqref{eq_dist}. 
%
In the conflict graph model, a schedule $\x$ is said to be feasible if it is an independent set of the underlying graph $G(V,E)$. Using this observation, the term $\i(\x \text{ is feasible})$ in \eqref{eq_dist} can be factorized in two ways as given below.\\
\emph{1. Edge-centric factorization:} This factorization ensures that for each edge, not more than one of its end vertices are active.
\begin{align*}
\i(\x \text{ is feasible})&= \prod_{(i,j) \in E} \i(x_i x_j =0).
\end{align*}

\noindent \emph{2. Vertex-centric factorization:} In this factorization, whenever a vertex is active, it ensures that all its neighbours are inactive.
\begin{align*}
\i(\x \text{ is feasible})&= \prod_{j \in V} f_j({\x^{(j)}}),
\end{align*}
where $\x^{(j)}$ is the set of activation status of node $j$ and its neighbours, and $f_j$ is given by
\begin{align*}
f_j(\x^{(j)})= \begin{cases}
 1,  & \text{if } x_j = 0,\\
1, & \text{if }  x_j=1 \text{ and } \{x_k\}_{k \in \N_j \setminus \{j\}}=0, \\
  0, &  \text{otherwise}.
\end{cases}
\end{align*}
In \cite{bethe_jshin}, the authors use the edge-centric factorization, which cannot be directly used to capture the SINR model. Hence we use a more general factorization \eqref{eq_local_factors}, which reduces to the vertex-centric factorization for the conflict graph case. 

\emph{Remark:} 
The existing theory on the Bethe approximation is not sufficient to conclude or prove that one of the factorizations is always better than the other. However, the following can be said under some special cases.

\begin{itemize}
\item For the conflict graph model, the formula derived in \cite{bethe_jshin} is provably exact if the conflict graph is a tree. Hence, if there is prior knowledge that the conflict graph has tree topology, the formula in \cite{bethe_jshin} should be used. 

\item From the simulations that we have conducted, we have the following observation for topologies with loops. For the formula proposed in \cite{bethe_jshin}, we observed that the achieved service rate is less than the target service rate even for small load. However, for the formula obtained using our approach, the achieved service rate is more than the target service rate at smaller loads. Hence, if we have prior knowledge that we are operating in the low load regime, empirical evidence suggests that it is beneficial to use the approach proposed in our paper.
\end{itemize}

\section{Utility Maximization} \label{util_max}
In this section, we consider the utility maximization problem and provide an approximation algorithm to solve the problem in a distributed manner. The problem is defined as follows. Suppose each link $i$ in the network is associated with a concave utility function of its service rate $U_i : [0,1] \rightarrow \R_+$. Our objective is to find the service rates that maximize the system wide utility, \ie, 
\begin{align}
\max_{y \in \Lambda} \sum_{i=1}^N U_i(y_i), \label{eq_util_main}
\end{align}
and subsequently compute the global fugacities that correspond to these optimal service rates.

In \cite{libin}, an iterative algorithm to update the global fugacities is proposed. However, each iteration of the algorithm requires an underlying slowly mixing Markov chain to reach steady state. Hence it suffers from impractically slow convergence to the optimal fugacities. To address this issue, we propose an iterative algorithm which updates the local fugacities instead of directly updating the global fugacities. These local fugacitiy updates are computationally simple and do not require any Markov chain to convergence. These local fugacities are then used to obtain the approximate global fugacities using \eqref{eq_lf_gf0} proposed in Section \ref{local_gibbs}. 

Each link in the network executes \emph{Algorithm 2} in parallel. The algorithm involves solving an one dimensional optimization problem \eqref{eq_one_dim} related to the original optimization problem \eqref{eq_util_main}. Here $\theta >0,$ is a parameter of the algorithm that can be tuned. The solution of this optimization problem $s_j(t)$ is used in the subsequent steps of the algorithm to update the local fugacities \eqref{eq_update} and the global fugacities. We will later show that the update equation \eqref{eq_update} is inspired by a subgradient descent algorithm for a related optimization problem. The term $\alpha(t)$ in  \eqref{eq_update} corresponds to the step-size of the subgradient descent algorithm. Any standard step-sizes that satisfy the convergence criteria for a subgradient descent method can be used \cite[Chapter 2]{book_subgradient}. A typical example is $\alpha(t) = \frac{1}{t}$.

\noindent\rule[0.5ex]{\linewidth}{0.5pt}
\textbf{Algorithm 2: Local utility maximization at link $j$}\\
\noindent\rule[0.5ex]{\linewidth}{0.5pt}
\begin{enumerate}
\item At $t=0$, initialize $\beta_{jk}(t)=0,  \; \forall k \in \nj$.\\

\emph{Computing global fugacities:}\\

\item  From each neighbour $k \in \N_j$, obtain the local fugacity $\lf_{kj}(t)$. 

\item Compute $s_j(t)$ as
\begin{align}
s_j(t)= \arg \max_{q \in [0,1]} \theta U_j(q) - q \sum_{k \in \nj} \beta_{kj}(t). \label{eq_one_dim}
\end{align}

\item Compute the approximate global fugacity $\lt_j(t)$ from local fugacities $\left( \lf_{kj}(t), \; k \in \nj \right)$ and $s_j(t)$ using \eqref{eq_lf_gf0}.\\

\emph{Updating the local fugacities:}\\
\item Consider the distribution
\begin{align*}
\b_j(\x^{(j)};t) =  {Z_j^{-1}} \exp \Big(\sum_{k \in \N_j} x_k \beta_{jk}(t) \Big),  \forall \x^{(j)} \in \I_j,
\end{align*}
  and for each $k \in \N_j$, let $m_{jk}(t):=\sum_{\x^{(j)} : x_k=1} \b_j(\x^{(j)};t)$ represent the marginal probability corresponding to $x_k$ under this distribution. 
\item For each $k \in \nj$, obtain $s_k(t)$ computed at the neighbour $k$, and update the local fugacity at $j$ using
\begin{align}
\beta_{jk}(t+1)= \beta_{jk}(t) + \alpha(t)\left(s_k(t)-m_{jk}(t)\right). \label{eq_update}
\end{align}
\end{enumerate}
\noindent\rule[0.5ex]{\linewidth}{0.5pt}

\emph{Complexity:}
The computational complexity of the above algorithm is $O(2^{|\N_j|})$, which depends only on the size of the local neighbourhood, and is independent of the total size of the network.

\emph{Accuracy:}
The accuracy of the above algorithm can be understood by splitting the error into two parts. The first part is the error involved in estimating the optimal service rates in \eqref{eq_util_main}. The second part is the error involved in estimating the global fugacities corresponding to these service rates. We characterize the first part of the error in Theorem \ref{thm_util}, which states that the gap to the optimum utility is $O(\frac{1}{\theta})$.  In other words, the local algorithm provides a good estimate of the optimal service rates when $\theta$ is large. The second part of the error depends only on the accuracy of the Bethe approximation. The Bethe error is evidenced to be reasonably small in many applications \cite{bethe_emprical}.

\begin{theorem} \label{thm_util}
In Algorithm 2, the service rates $s(t)=[s_j(t)]_{j=1}^N$ \eqref{eq_one_dim} converge to some $\underline{s}=[\underline{s}_j]_{j=1}^N \in [0,1]^{\N}$, such that the limit $\underline{s}$ satisfies
\begin{align}
\sum_{j=1}^N U_j(\underline{s}_j) \geq \max_{y \in \Lambda} &\sum_{j=1}^N U_j(y_j) - \frac{\sum_j \log |\mathcal{I}_j|}{\theta}. \label{eq_thm_util}
\end{align}
\end{theorem}
\begin{proof}
We prove this theorem in the following steps:
\begin{itemize}
\item[1.] First, we define a new optimization problem \eqref{eq_util_bethe} that is related to the original utility maximization problem \eqref{eq_util_main}.
\item[2.] Then, we prove that the proposed local utility maximization algorithm corresponds to a subgradient descent algorithm for the new optimization problem \eqref{eq_util_bethe}.
\item[3.] Finally, we complete the proof by showing that the solution of the new optimization problem \eqref{eq_util_bethe} satisfies the inequality \eqref{eq_thm_util} stated in this theorem.
\end{itemize} 
\noindent  \emph{Step 1:}
Firstly, we relax the constraints of \eqref{eq_util_main} by replacing the actual capacity region $\Lambda$, with the Bethe approximated capacity region defined below:
\begin{align}
\Lambda_B:=\{y \in [0,1]^{\N} \; | \; y^{(j)} \in \Lambda_j, \; \forall j \in \N\}. \label{eq_bethe_cap}
\end{align}
From Section \ref{local_gibbs}, recall that $y^{(j)}=[y_k]_{k \in \N_j}$ is the local service rate vector, and $\Lambda_j$ defined in \eqref{eq_loc_cap} is the local capacity region which is nothing but the convex hull of the local feasible schedules $\I_j$.  In other words, a local service rate vector $y^{(j)}$ belongs to $\Lambda_j$, if and only if there exists a distribution $\b_j$ on the local feasible schedules $\I_j$, that supports the service rates $y^{(j)}$, \ie, $ y_k= \sum_{\{\x^{(j)} \in \I_j | x_k=1\}} \b_j(\x^{(j)}), \; \;  k \in \N_j$. As this definition of $\Lambda_B$ imposes only local feasibility of service rates, it is easy to argue that the Bethe capacity region is a relaxation of the actual capacity region, \ie, $\Lambda \subseteq \Lambda_B$.


Secondly, we scale the objective function  \eqref{eq_util_main} by a factor $\theta$, and add local entropy terms. The resulting new optimization problem is given by
\begin{align}
&\max_{y, \{\b_j\}}  \; \; \theta \;\sum_{j=1}^N U_j(y_j) + \sum_{j=1}^N H(\b_j), \; \; \; \; \text{subject to } \label{eq_util_bethe} \\
&y \in [0,1]^{\N}; \; \; \b_j(\x^{(j)}) \geq 0, \; \; \x^{(j)} \in \mathcal{I}_j, \; j=1 \dots N ; \; \nonumber\\
 &\sum_{\x^{(j)} \in \mathcal{I}_j} \b_j(\x^{(j)})=1,\; \;  j=1 \dots N; \nonumber \\
& \;  y_k= \sum_{\{\x^{(j)} \in \I_j | x_k=1\}} \b_j(\x^{(j)}), \; \;  k \in \N_j,\; j=1 \dots N . \label{eq_util_constraint}
\end{align}
Here, the constraint set is specified by explicitly expanding the definition of $\Lambda_j$ involved in description of $\Lambda_B$ \eqref{eq_bethe_cap}. For a large $\theta$, this new objective function \eqref{eq_util_bethe} closely approximates the original objective \eqref{eq_util_main} since the entropy is bounded from above and below. The advantage of defining this new optimization problem is that it is amenable to a distributed solution \cite{libin}. In particular, our local algorithm solves this problem in a distributed fashion.

\noindent  \emph{Step 2:} In the following Lemma, we show that Algorithm 2 solves the new optimization problem \eqref{eq_util_bethe}. 
\begin{lemma} \label{lemma_sub_grad}
Let $(\underline{s}, \{\underline{\b}_j\})$ be the solution of the optimization problem \eqref{eq_util_bethe}. Then the service rates $\{s_j(t)\}$ \eqref{eq_one_dim} in Algorithm 2 converge to the limit $\underline{s}=\{\underline{s}_j\}$.
\end{lemma}
\begin{proof}
The outline of the proof is to show Algorithm 2 corresponds to the dual subgradient method for \eqref{eq_util_bethe}. The proof is provided in Appendix \ref{proof_sub_grad}. 
\end{proof}

\noindent {Step 3:}
We now show that the solution of the optimization problem \eqref{eq_util_bethe} achieves the performance guarantee claimed in \eqref{eq_thm_util}. In other words, we show that $\underline{s}=\{\underline{s}_j\}$ satisfies \eqref{eq_thm_util}.

Recollect that we have added the local entropy terms to the actual objective function to obtain a new problem \eqref{eq_util_bethe}. Here, we shall evaluate the effect of these entropy terms on the optimal utility. To that end, let us consider $(\underline{s}, \{\underline{\b}_j\})$, the solution to the maximization problem \eqref{eq_util_bethe}.  Let us also consider the following optimization problem which does not have the local entropy terms. $s^*:= \arg \max_{y \in \Lambda_B} \sum_j U_j(y_j).$
Note that $s^*$ could be different from $\underline{s}$. Also, from the definition of $\Lambda_B$, for every $y \in \Lambda_B$ there exists some $\{\b_j\}$ such that $(y, \{\b_j\})$ is feasible for the problem \eqref{eq_util_bethe}. Hence for $s^* \in \Lambda_B$, there exists some $\{\b_j^*\}$ such that $(s^*, \{\b_j^*\})$ is feasible for \eqref{eq_util_bethe}. Then,
\begin{align}
\theta \sum_j U_j(s^*_j) &\leq \theta \sum_j U_j(s_j^*) + \sum_j H(\b^*_j), \nonumber\\
&\overset{(a)}{\leq} \theta \sum_j U_j(\underline{s}_j) + \sum_j H(\underline{\b}_j), \nonumber\\
&\overset{(b)}{\leq} \theta \sum_j U_j(\underline{s}_j) + \sum_j \log |\mathcal{I}_j|, \label{eq_entropy1}
\end{align}
where $(a)$ follows from the fact that $(\underline{s}, \{\underline{\b}_j\})$ maximizes \eqref{eq_util_bethe}, $(b)$ follows from the fact that entropy $H(\b_j) \leq \log |\mathcal{I}_j|$. Next, using the fact that $\Lambda \subseteq \Lambda_B$, we have
\begin{align}
\max_{y \in \Lambda} \sum_j U_j(y_j) &\leq \max_{y \in \Lambda_B} \sum_j U_j(y_j) = \sum_j U_j(s^*_j). \label{eq_entropy2}
\end{align}

\noindent Using \eqref{eq_entropy2} in \eqref{eq_entropy1} completes the proof of Theorem \ref{thm_util}.
\end{proof}

We conclude this section by contrasting our utility maximization algorithm with a prior work \cite{libin} which is also a subgradient descent based algorithm. The key difference between these works is in the complexity involved in computing the subgradients. Specifically, in our algorithm, the complexity of computing the subgradient is independent of the total size of the network. On the other hand, it could be exponentially large in the network size, for the algorithm given in \cite{libin}.

\section{Numerical results} \label{simulations}
\begin{figure*} 
\begin{subfigure}{.33 \textwidth}
\centering
\includegraphics[scale=0.28]{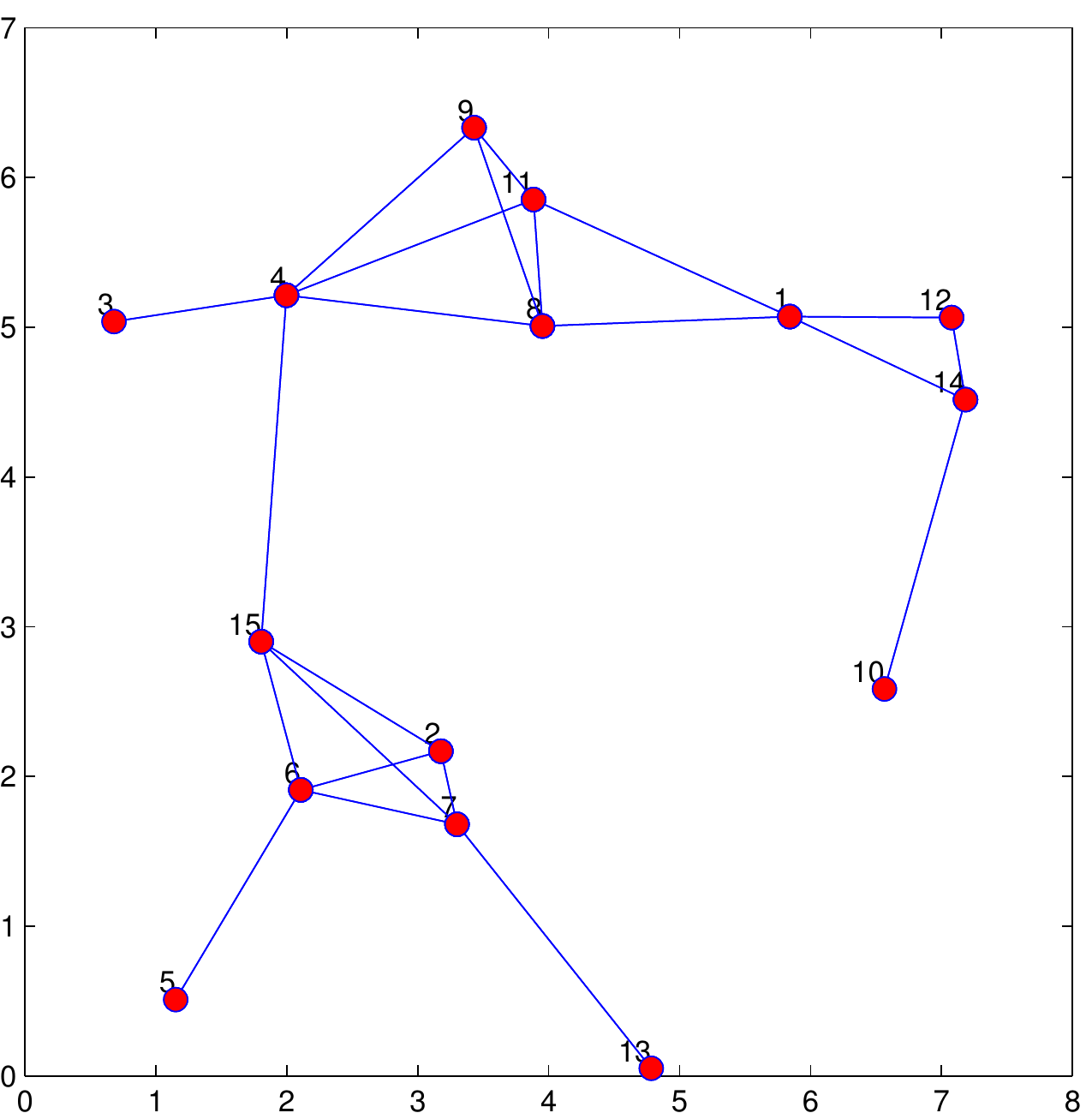} 
\caption{15-link random topology} 
\label{fig_rg15}
\end{subfigure}
\begin{subfigure}{.33 \textwidth}
\centering
\includegraphics[scale=0.34]{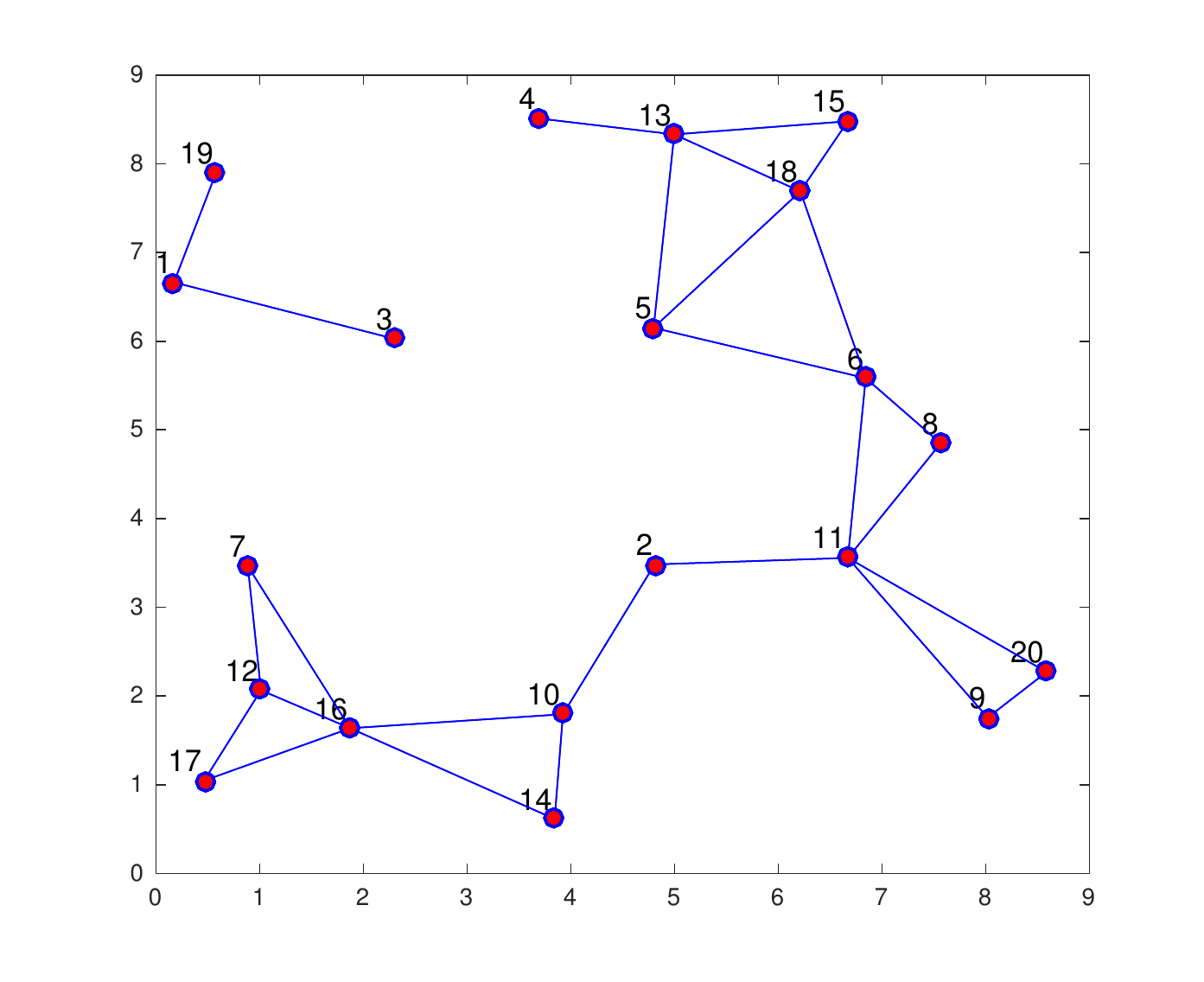} 
\caption{20-link random topology} 
\label{fig_rg20}
\end{subfigure}
\begin{subfigure}{.33 \textwidth}
\centering
\includegraphics[scale=0.34]{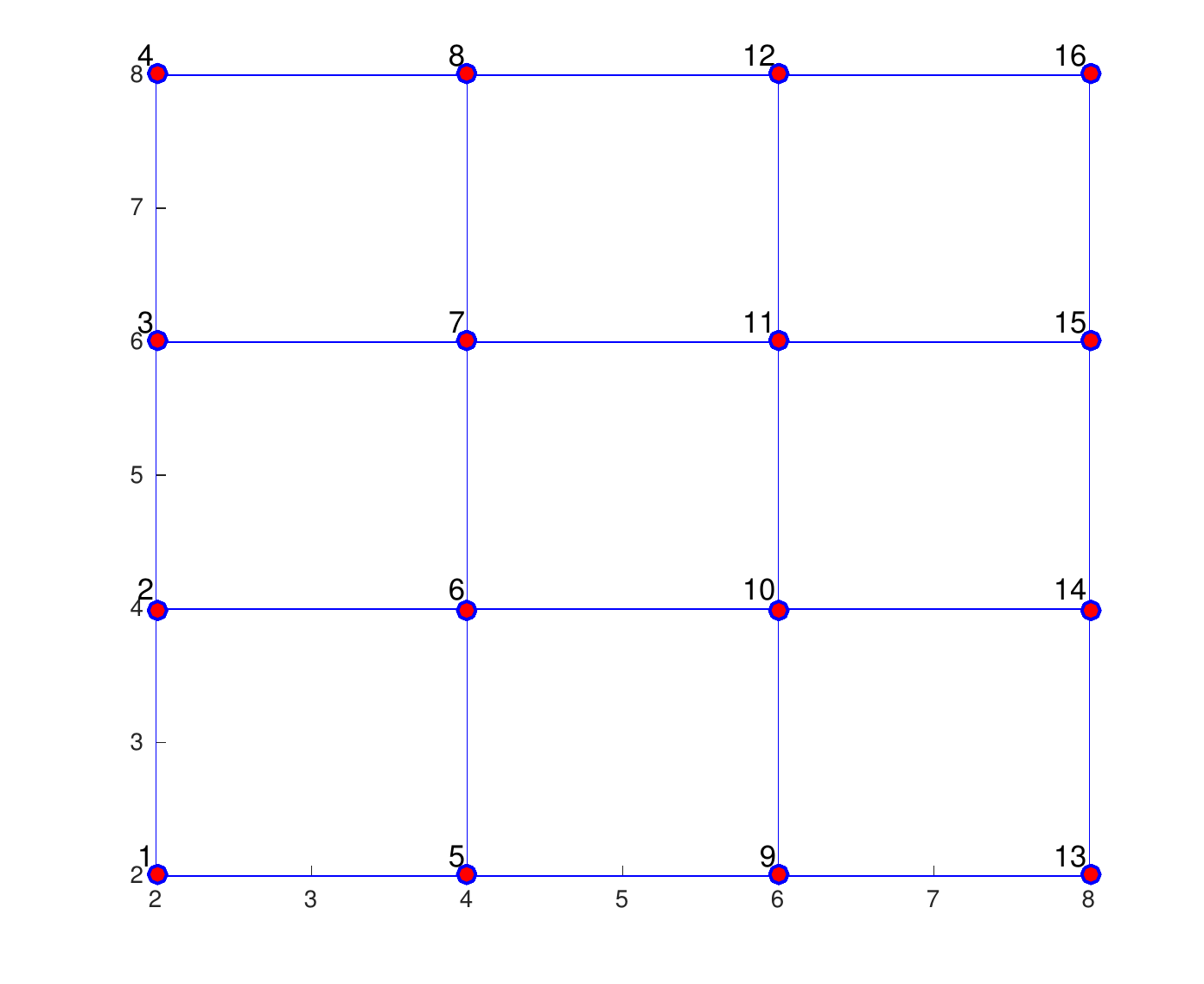}
\caption{4 $\times$ 4 grid topology}
\label{fig_grid_graph}
\end{subfigure}
\caption{Illustration of some interference graphs used: Each vertex represents a link in the network. An edge is present if two links are in the interference range.}
\label{fig_all}
\end{figure*}

\begin{figure*}
\begin{subfigure}{.5 \textwidth}
\centering
\includegraphics[scale=0.38]{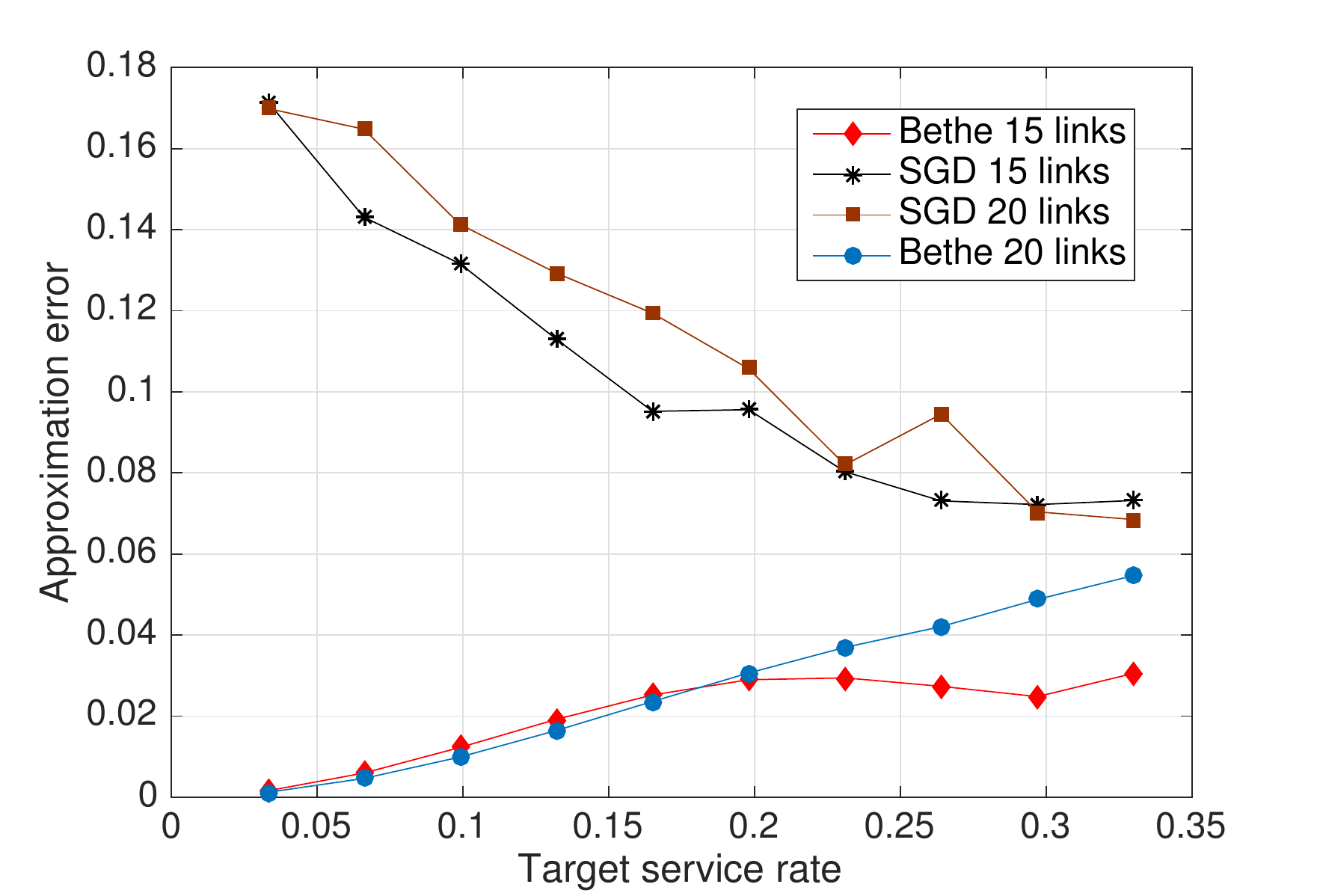} 
\caption{Error as a function of target service rate for the random topologies}
\label{fig_be_load_20}
\end{subfigure}
\begin{subfigure}{.5 \textwidth}
\centering
\includegraphics[scale=0.4]{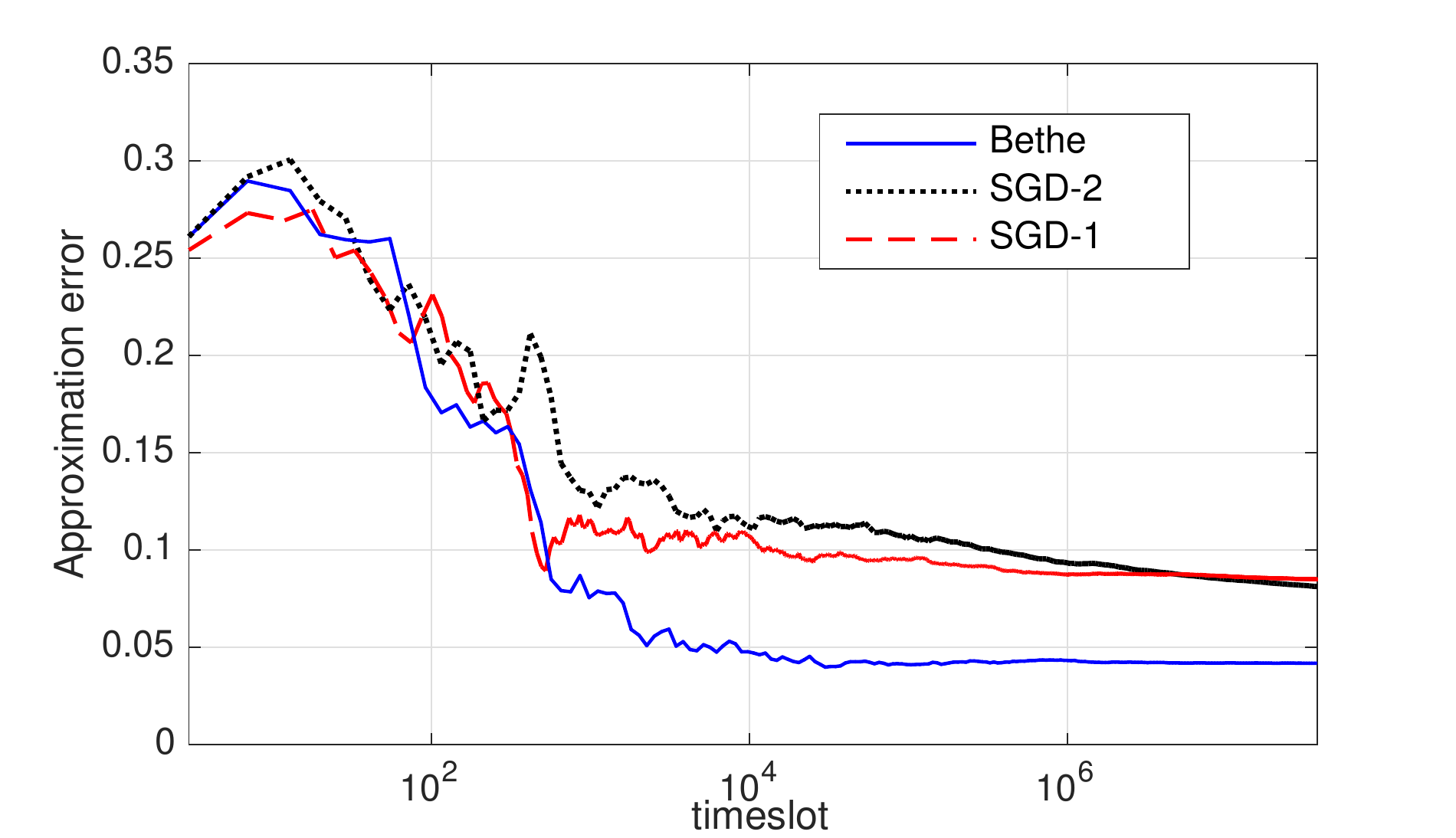}
\caption{Error as a function of time for 20-link random graph}
\label{fig_time}
\end{subfigure}
\caption{Comparision of errors due to SGD and our Bethe approximation based algorithms for random topologies in Figure \ref{fig_all}}
\end{figure*}

\begin{figure*}
\begin{minipage}{.5 \textwidth}
\centering
\includegraphics[scale=0.38]{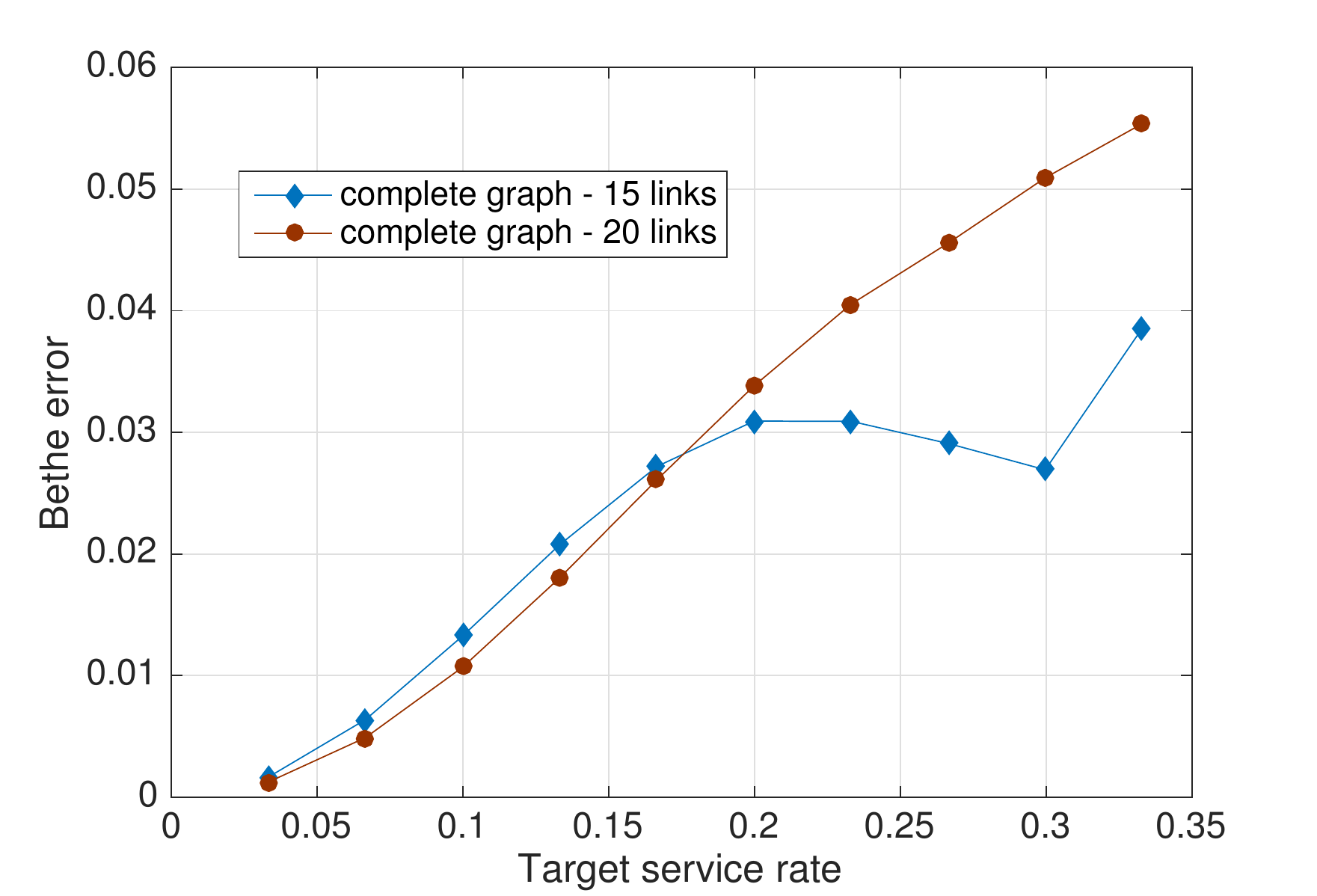} 
\caption{Bethe error for complete graph topology}
\label{fig_completegraph_error}
\end{minipage}
\begin{minipage}{.5 \textwidth}
\centering
\includegraphics[scale=0.4]{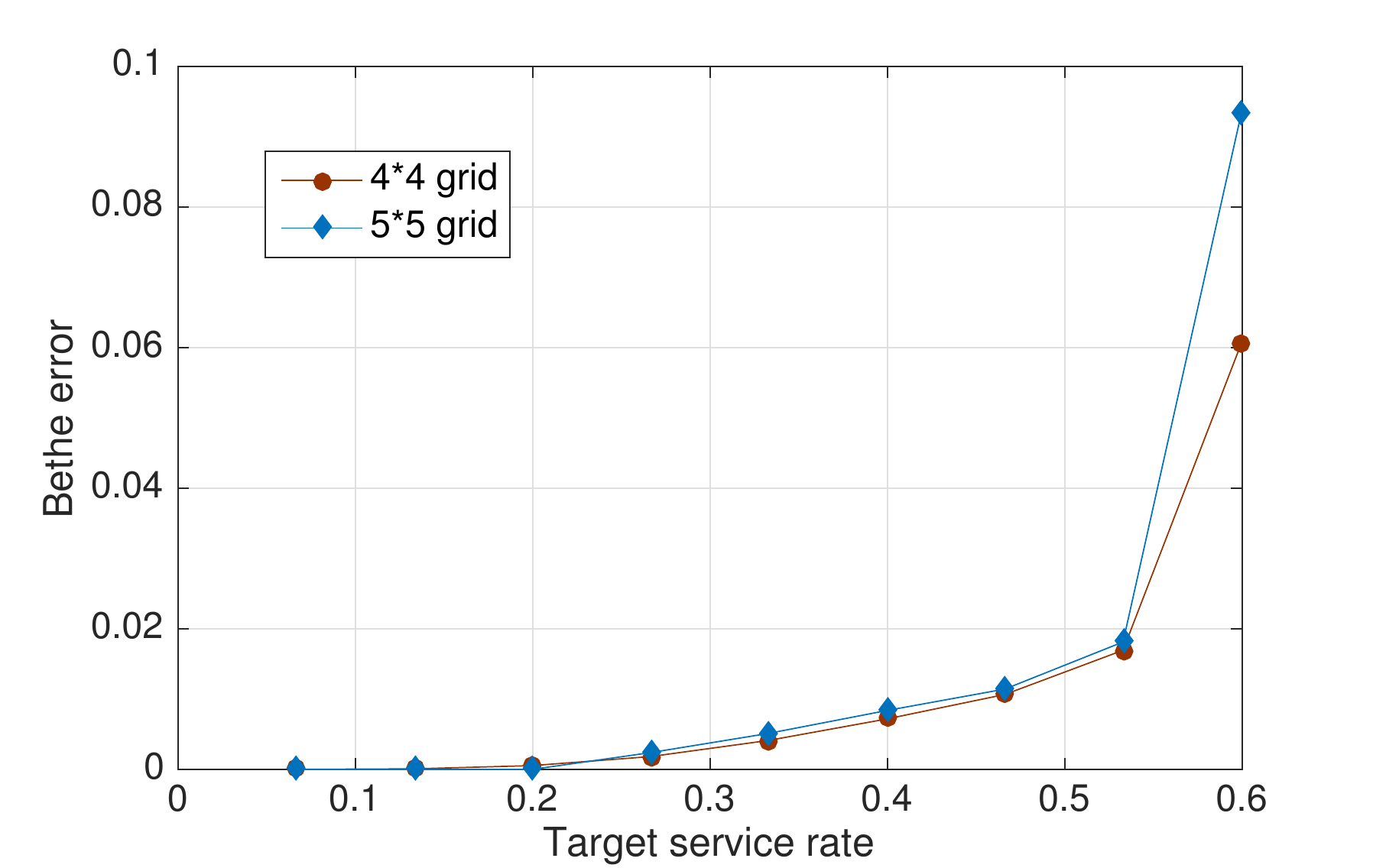}
\caption{Bethe error for grid topology}
\label{fig_gridgraph_error}
\end{minipage}
\end{figure*}

\begin{figure}
\centering
\includegraphics[scale=0.46]{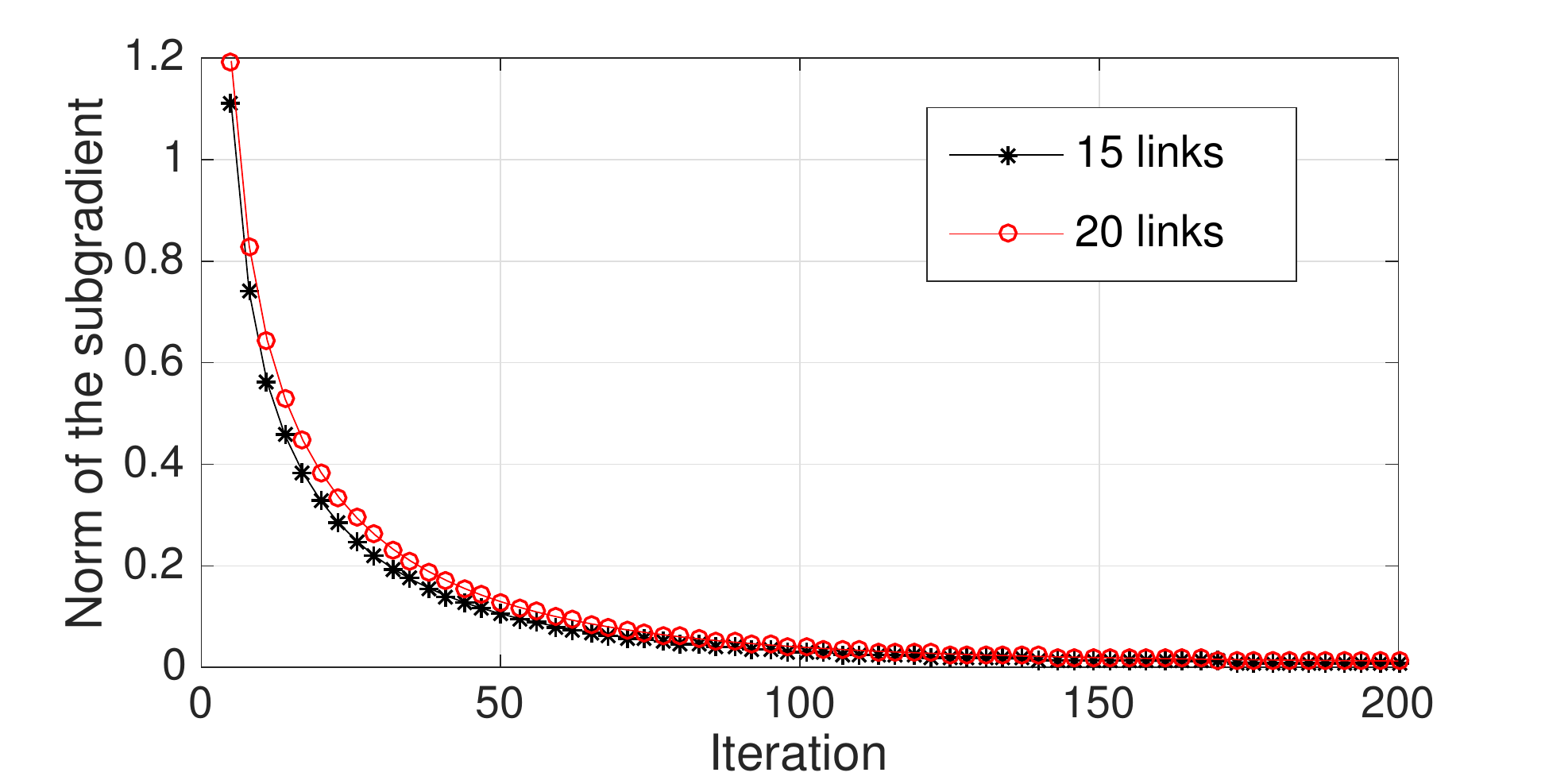} 
\caption{Convergence of local utility maximization algorithm}
\label{fig_umax}
\end{figure}
In this section, we present simulations to evaluate the performance of our algorithms. We consider three topologies namely random graphs, complete graph, and grid graph.
\subsection{Bethe error}
\emph{Simulation setting for random graphs:} We generate a spatial random network by uniformly placing the transmitter nodes on a two dimensional square plane of length $8$.  Each transmitter is associated with a receiver at a distance of $0.5$ in a random direction. The path loss exponent $\alpha$ is set to $3$, the close-in radius $R_I$ is set to $2.4$, and the threshold SINR is set to $15$ dB. The transmit power of the links is set to $1$. We consider equal service rate requirements for all the links. For a given target service rate vector $s^t=[s_i^t]_{i=1}^N$,  we define the Bethe approximation error as 
\begin{align*}
e(s^t)= \frac{\sum_{i=1}^N  |s^t_i - s^a_i|}{N},
\end{align*}
where $s^a=[s_i^a]_{i=1}^N$ are the service rates that can be supported by the approximated fugacities $\{\lt\}_{i=1}^N$. 

\emph{Bethe error as function of target service rate:}
We generated two random networks of sizes $15$ and $20$ links (shown in Figure \ref{fig_rg15}, \ref{fig_rg20}). We compared the error of our approximation algorithm with the residual error of the SGD algorithm \cite{libin} after running the SGD for $10^{8}$ time-slots of the CSMA algorithm. In Figure \ref{fig_be_load_20}, we plot the approximation error as a function of the target service rate. Further, for the considered SINR constraints, we numerically observed that for both the random topologies in Figure \ref{fig_all}, CSMA can support till a  service rate of $0.33$ for all the links. Hence, we varied the target service rate from $0$ to $0.33$. It can be observed from Figure \ref{fig_be_load_20}, that for practical time-scales, our algorithm results in better accuracy than the SGD for all the target service rates.

\emph{Error as a function of time:} In Figure \ref{fig_time}, for a fixed target service rate of  $0.25$, we plot the approximation error as a function of time, by using the time-averaged service rates observed from CSMA algorithm. Here, our local Gibbs algorithm computes the Bethe approximated fugacities, and uses these static fugacities (\ie, they are not adapted during the algorithm) in the CSMA algorithm. The SGD algorithm starts with some initial fugacities, and adapts the fugacities by observing the corresponding service rates. We simulated two versions of the SGD algorithm (SGD-1, SGD-2) proposed in \cite{libin,libin_arxiv}, whose details are as follows: The update rule of SGD algorithm has two functions to be chosen, namely update interval $T(j)$, and step size $\alpha(j)$, for the $j^{th}$ iteration of the gradient descent. The update rule for SGD-1 \cite[Section II-D]{libin} is given by $\alpha(j)=\frac{1}{(j+2)\log(j+2)}$, $T(j)=j+2$. The update rule for SGD-2 \cite[Scheduling Algorithm 1]{libin_arxiv} is given by $\alpha(j)=\frac{1}{j}$, $T(j)=\exp(\sqrt{j})$. 

Although the SGD algorithm will eventually converge to the exact fugacities, from Figure \ref{fig_time}, it can be observed that for practical time-scales of the order of $10^8$ time slots, the residual error is rather large compared to our approximation algorithm. This is because, the CSMA Markov chain has to mix only one time for the Bethe approximation based approach, as the fugacities are static. However, in the SGD based approach, for every update in the fugacities, the Markov chain tries to mix to a new steady state distribution.

\emph{Complete graph and Grid graph:}
Here, we consider two complete graph topologies with sizes $15$ and $20$, and two grid topologies of sizes $16$ and $25$. A $4 \times 4$ grid topology is illustrated in Figure \ref{fig_grid_graph}.  The range of the target service rates is chosen by numerically observing the maximum supportable service rates for the respective topologies. We plot the Bethe error for these two topologies as a function of the target service rates in Figure \ref{fig_completegraph_error}, \ref{fig_gridgraph_error}. As seen from the plots, the error is considerably small for both the topologies.

\subsection{Convergence of Utility maximization algorithm}
Here, we consider the case of computing the fugacities for the proportional fairness utility setting, \ie, $U_i(s_i)= \log s_i$ for all the links. We update the fugacities using Algorithm 2 proposed in Section \ref{util_max}. In Figure \ref{fig_umax}, we plot the norm of the subgradient corresponding to the  maximization problem \eqref{eq_util_bethe}, which indicates the convergence of the algorithm. The convergence is plotted for the two random topologies of size $15$ and $20$. It can be observed that the proposed algorithm converges within 200 iterations.  

\section{Conclusions} \label{conc}
We considered the adaptive CSMA algorithm under the SINR interference model, which is known to be throughput optimal. Under this model, we first proposed a distributed algorithm, namely the local Gibbsian method to efficiently estimate the fugacities, for a given service rate requirements. The convergence rate and the complexity of the proposed algorithm depend only on the maximum size of a link's neighbourhood. We proved that our approximation corresponds exactly to performing the well known Bethe approximation to the global Gibbsian problem. We also proposed an approximation algorithm to estimate the fugacities under a utility maximization framework. Our numerical results indicate that the proposed approximation algorithms can lead to a good degree of accuracy, and improve the convergence time by a few orders of magnitude, compared to the existing stochastic gradient descent methods. 

\section*{Acknowledgement}
 The authors express their sincere gratitude to Dr. Pascal Vontobel for his valuable comments on this work.

\section{Appendix}
\subsection{Proof of Lemma \ref{lemma_max_entropy}} \label{proof_max_entropy}

\begin{figure*}
\begin{center}
  \vspace{-4mm}
\begin{tikzpicture}

\fill (1,2)circle(0.06cm);
\fill (2,2)circle(0.06cm);
\fill (3,2)circle(0.06cm);

\node at (.9,2.2) {\begin{scriptsize}1\end{scriptsize}};

\node at (2,2.2) {\begin{scriptsize}2\end{scriptsize}};

\node at (3.1,2.2) {\begin{scriptsize}3\end{scriptsize}};

\draw (1,2)-- (2,2) -- (3,2);  

\end{tikzpicture}
\caption{Illustration of a 3-node interference graph.} 
  \label{fig_cliq}
  \end{center}
  \vspace{-2mm}
  \end{figure*}
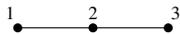 
  
\begin{table*}[t]
\centering
\hspace{15mm}   \bf{Variable marginals:} $b_1(x_1)$ \; $b_2(x_2)$ \; $b_3(x_3)$\
  \newline
  
\begin{tabular}{l|l|l}
$b_1(0)= 1-y_1$ & $b_2(0)= 1-y_2$ & $b_3(0)= 1-y_3$\\
$b_1(1)= y_1$ & $b_2(1)= y_2$ & $b_3(1)= y_3$
\end{tabular}

\end{table*}

\begin{table*}[t]
\centering
\hspace{15mm}  {\bf Factor marginals:} $\b_1(x_1,x_2)$ \; $\b_2(x_1,x_2,x_3)$ \; $\b_3(x_2,x_3)$\
\newline
 
\begin{tabular}{l|l|l}
$\b_1(0,0)= 1-y_1-y_2+z_1(1)$   	& $\b_2(0,0,0)= 1- \sum_{i=1}^3 y_i+\sum_{k=1}^3 z_2(k)$ 	& $\b_3(0,0)= 1-y_2-y_3+z_3(1)$\\
$\b_1(0,1)= y_2-z_1(1)$ 	         	& $\b_2(0,0,1)= y_3-z_2(1)-z_2(2)$ 		& $\b_3(0,1)= y_3-z_3(1)$\\
$\b_1(1,0)= y_1-z_1(1)$ 					& $\b_2(0,1,0)= y_2-z_2(1)-z_2(3)$ 	& $\b_3(1,0)= y_2-z_3(1)$\\
$\b_1(1,1)= z_1(1)$ 						& $\b_2(0,1,1)= z_2(1)$ 			& $\b_3(1,1)= z_3(1)$\\
													& $\b_2(1,0,0)= y_1-z_2(2)-z_2(3)$ 		& \\
 													& $\b_2(1,0,1)= z_2(2)$ 			& \\
 													& $\b_2(1,1,0)= z_2(3)$ 		& \\
													& $\b_2(1,1,1)= 0$ 			& \\
\end{tabular}
\caption{Illustration of the variable transformations used in the proof of Lemma \ref{lemma_max_entropy}, for the 3-node topology given in Figure \ref{fig_cliq}}
\label{table_factor}
\end{table*}
Let us consider the Bethe optimization for the BFE defined in \eqref{eq_bfe}:
\begin{align}
\underset{\{\b_i\}_{i=1}^N,  \{b_i\}_{i=1}^N} {\arg \min}\; \;  F_B\left( \{\b_i\}_{i=1}^N,  \{b_i\}_{i=1}^N \right), \text{  subject to}\label{eq_constraint}
\end{align}
\vspace{-2mm}
\begin{align}
\sum\limits_{\x^{(i)} \setminus \{x_j \}} \b_i(\x^{(i)})&= b_j(x_j), \; i \in \N,  j \in \N_i, x_j =1, \label{eq_consistency} \\ 
 \sum\limits_{\x^{(i)} \in \I_i} \b_{i} (\x^{(i)})&=1, \hspace{0.4cm} i = 1 \dots N,\label{eq_sum1}\\
 b_{i} (1)+ b_i(0)&=1, \hspace{0.4cm} i = 1 \dots N. \label{eq_sum2}
\end{align}
As all the constraints of \eqref{eq_constraint} are linear, they can be eliminated by suitable variable transformation to get an equivalent unconstrained problem \cite[chapter 10]{boyd}. Let us now look at the optimization variables in \eqref{eq_constraint}. If we consider the distribution $b_i(\cdot)$, there are $2$ optimization variables associated with it, namely $b_i(0)$ and $b_i(1)$. If we consider the distribution $\b_i(\cdot)$, there is an optimization variable corresponding to each argument $\rho \in \I_i$, \ie, there are $|\I_i|$ variables associated with it, namely $\{ \b_i(\rho) \}_{\rho \in \I_i}$. For convenience, we split the arguments $\rho \in \I_i$ into two sets. We use $A$ to denote the set of arguments in which there is at most one non-zero element, \ie, $A= \{ \rho =[\rho_j]_{j \in \N_i} \in \I_i \; | \; \sum_{j \in \ni} \rho_j \leq 1 \},$ and use $B$ to denote the set of all the other arguments in $\I_i$. Later, we use the definition of these two sets, to propose some variable transformations that eliminate the equality constraints.

\emph{Example:} Before we proceed further, let us consider an example which will be used to illustrate the variable transformations. Consider an interference graph with 3-node line topology as shown in Figure \ref{fig_cliq}. For the ease of illustration, let us assume that the SINR constraints are such that $(1,1,1)$ is the only infeasible local schedule. In other words, $\I_1=\{0,1\}^2$, $\I_2=\{0,1\}^3 \setminus (1,1,1)$, $\I_3=\{0,1\}^2$.

Then for this example, let us consider node $2$, and split the arguments of $\b_2(\rho)$ into sets $A$ and $B$. As defined earlier, the set of arguments which contain at most one non-zero element is given by $A=\{(0,0,0), (1,0,0), (0,1,0), (0,0,1)\}$. The other set $B$ containing feasible arguments with more than one non-zero element is given by $B=\{ (0,1,1), (1,0,1), (1,1,0)\}$. Similarly, if we consider node $1$, the corresponding sets will be $A=\{(0,0), (0,1), (1,0)\}$ and $B=\{(1,1)\}$.

\emph{Variable transformations:}
Let us now perform the following variable transformations to eliminate the constraint equations, and thereby obtain an equivalent unconstrained optimization problem. For each $i \in \N$,
\begin{itemize}
\item[(i)] Replace the variable $b_i(1)$ by a new variable $y_i$.  Hence, by expressing $b_i(0)=1-y_i$, we can eliminate the constraint $b_i(0)+b_i(1)=1$ in \eqref{eq_sum2}.
\item[(ii)]  Let a new variable $\z_i =[z_i(k)]_{k=1}^{|B|}$ (which is a vector of length $|B|$) replace the optimization variables corresponding to the arguments in $B$.  For example, if we consider the node $1$, the set $B=\{(1,1)\}$ has only one element. We use $\z_1=(z_1(1))$, a vector of length one to replace $\b_1(1,1)$. Similarly if we consider node $2$, its set $B=\{ (0,1,1), (1,0,1), (1,1,0)\}$ has $3$ elements. We use the vector $\z_2=(z_2(1), z_2(2),z_2(3))$ to replace $\{\b_2(0,1,1), \b_2(1,0,1), \b_2(1,1,0)\}$ respectively. This transformation is shown in Table \ref{table_factor}.

\item[(iii)]  If an argument $\rho \in A$ has the non-zero element in $j$th position (\ie, $\rho_j$  $=$ $1$), then its corresponding variable $\b_i(\rho)$ can be expressed only in terms of the vector $\z_i$  and variable $y_j$. This transformation is to replace the local consistency condition in \eqref{eq_consistency}. For example, if we consider the 3-node graph, the argument of the probability $\b_2(1,0,0)$ has its non zero element at position $j=1$.  Now consider the local consistency constraint \eqref{eq_consistency} for $i=2, j=1$ to obtain $\b_2(1,0,0) + \b_2(1,0,1) + \b_2(1, 1, 0) = b_1(1).$ Then to capture the above constraint, the transformation $\b_2(1,0,0)= y_1-z_2(2)-z_2(3)$ can be used. Thus $\b_2(1,0,0)$ is expressed only in terms of the elements of the vector $\z_2$, and variable $y_1$. It is shown in Table \ref{table_factor}.
\item[(iv)] The variable $\b_i(\rho)$ corresponding to $\rho = (0,0, \dots, 0)$ can be expressed in terms of the vector $\z_i$ and $(y_j , j \in \ni )$. This is done to eliminate \eqref{eq_sum1}.  For example, at node $1$, the constraint $\b_1(0,0)+\b_1(0,1) + \b_1(1,0) +\b_1(1,1)=1$ can be eliminated by expressing $\b_1(0,0)$ in terms of the vector $\z_1$ and $(y_1,y_2)$ as shown in Table \ref{table_factor}.
\end{itemize}
By following the above steps, the distribution $\b_i(\x^{(i)})$ can be expressed as a linear function of $(\z_i, \{y_j\}_{j \in \ni})$. Upon this transformation, the free energy $F_B(\bn,\bv)$ \eqref{eq_bfe} is expressed in terms of the new variables as
\begin{align}
&F_B \left( {\left\lbrace \z_{i} \right\rbrace}_{i \in \N}, {\left\lbrace y_{i} \right\rbrace}_{i \in \N} \right) \label{eq_fb_new}\\
&= \sum\limits_{i=1}^N \left[-(\ln \lambda_i) y_i - \H_i(\z_i, \{y_j\}_{j \in \ni})  + (d_i-1) H_i(y_i)\right], \nonumber
\end{align}
and the equivalent unconstrained optimization problem is to minimize $F_B \left({\left\lbrace \z_{i} \right\rbrace}_{i \in \N}, {\left\lbrace y_{i} \right\rbrace}_{i \in \N}\right)$ over the new variables. 
Technically, there should be additional inequality constraints on these new variables to impose positivity constraints described in \eqref{eq_bethe_min}. However, it can be shown that any stationary point of the BFE given in \eqref{eq_bfe} implicitly satisfies those positivity constraints  \cite[Remark 4.1 in Page 85]{book_martin}.

Now suppose $\left({\left\lbrace \z^*_{i} \right\rbrace}_{i \in \N}, {\left\lbrace y^*_{i} \right\rbrace}_{i \in \N}\right)$ is a stationary point of \eqref{eq_fb_new}. Then the gradient of $F_B ( {\left\lbrace \z_{i} \right\rbrace}_{i \in \N}, {\left\lbrace y_{i} \right\rbrace}_{i \in \N})$ at that stationary point should be zero. In particular, if we take the partial derivative of $F_B ( {\left\lbrace \z_{i} \right\rbrace}_{i \in \N}, {\left\lbrace y_{i} \right\rbrace}_{i \in \N})$ \eqref{eq_fb_new} with respect to the elements of the vector $\z_i=\{z_i(k)\}_{k=1}^{|B|}$,  all the terms in \eqref{eq_fb_new} other than the term corresponding to $\H_i$ vanishes. Therefore, setting  $\frac{\partial F_B \left( {\left\lbrace \z_{i} \right\rbrace}_{i \in \N}, {\left\lbrace y_{i} \right\rbrace}_{i \in \N} \right)}{\partial z_i(k)}= 0$, we have for $k=1 \text{ to } |B|$,
\begin{align}
 \left. \frac{\partial \H_i\left(\z_i, \{y_j\}_{j \in \ni}\right)}{\partial z_i(k)}\right|_{\left(\z_i, \{y_j\}_{j \in \ni}\right)=\left(\z_i^*, \{y^*_j\}_{j \in \ni}\right)}  &= 0. \label{eq_partialz}
\end{align}
Next, we consider the maximum entropy property \eqref{opt_local_entropy} stated in Lemma \ref{lemma_max_entropy}, and argue that it essentially boils to down to the above system of equations \eqref{eq_partialz}. Firstly, since the constraints of \eqref{opt_local_entropy} are same as the constraints required for the Bethe optimization problem \eqref{eq_constraint}, they can be eliminated using the variable transformations used in this proof. From \eqref{eq_dual}, it can be observed  that all the variable marginals are fixed at $\{b_j^*\}_{j \in \N_i}$, and the optimization \eqref{opt_local_entropy} is done only over the factor marginals $\b_i$. In terms of the transformed variables, it essentially boils down to maximizing the entropy $\H_i\left(\z_i, \{y_j\}_{j \in \ni}\right)$ subject to fixing the variable marginals at $\{y_j^*\}_{j \in \N_i}$, which is captured by \eqref{eq_partialz}. This observation essentially asserts that the factor marginals $\b_{i}^*$, and the variable marginals $\{b_j^*\}_{j \in \ni}$ corresponding to a stationary point $(\bn^*,\bv^*)$ are related by the maximum entropy problem defined in \eqref{opt_local_entropy}.

\subsection{Proof of Lemma \ref{lemma_fug}} \label{proof_fugacities}
This proof is a continuation of the proof of Lemma \ref{lemma_max_entropy}. By interpreting the factor and variable marginals $(\bn,\bv)$ as linear functions of the variables $\left({\left\lbrace \z_{i} \right\rbrace}_{i \in \N}, {\left\lbrace y_{i} \right\rbrace}_{i \in \N}\right)$, and setting the  partial derivative of $F_B\left({\left\lbrace \z_{i} \right\rbrace}_{i \in \N}, {\left\lbrace y_{i} \right\rbrace}_{i \in \N}\right)$ in \eqref{eq_bfe} with $y_i$ to zero, we obtain

\vspace{-3mm}
\begin{scriptsize}
\begin{align}
\ln \lambda_i&= \left.\frac{\partial \Big[(d_i-1) H_i(b_i) - \sum\limits_{j \in \ni} \H_j(\b_j)   \Big] }{\partial y_i}\right|_{\left(b_i, \{\b_j\}_{j\in \N_i}\right)=\left(b_i^*, \{\b_j^*\}_{j\in \N_i}\right)}. \label{eq_factor_global}
\end{align}
\end{scriptsize}
\vspace{-3mm}

Observe the following from steps (ii)-(iv) of the variable transformation in the proof of Lemma \ref{lemma_max_entropy}: The distribution $\b_j(\rho)$ depends on $y_i$, if only if the argument is either the all zero pattern, \ie, $(0,0, \dots, 0)$ or if $\rho_i=1$, is the only non zero element in that argument.  Let us denote this argument as $e^{(i)}:= (0, \dots,0,1,0, \dots, 0)$, where $1$ is in the $i^{th}$ position. Hence, only two terms of the entropy $\H_j(\b_j)$ depend on $y_i$.  (See Table \ref{table_factor}. For example, only two terms of $\b_2$, namely $\b_2(1,0,0)$ and $\b_2(0,0,0)$ depend on $y_1$.)

Using this observation in \eqref{eq_factor_global}, and simplifying gives us $\lambda_i= \left( \frac{\left(b^*_i(0)\right)^{d_i-1} \prod\limits_{j \in \ni} \b^*_j(e^{(i)})}{\left(b^*_i(1)\right)^{d_i-1} \prod\limits_{j \in \ni} \b^*_j(0, 0, 0, \dots, 0)}\right).$ Then we conclude the proof of this lemma by applying \eqref{eq_opt_factor_marg} from Lemma \ref{lemma_libin}, which gives us $\lambda_i= \left(\frac{b^*_i(0)}{b^*_i(1)}\right)^{d_i-1}\prod\limits_{j \in \ni} e^{v_{ji}}$.

\subsection{Sufficient condition for a stationary point of the BFE} \label{proof_complete}
The condition \eqref{eq_partialz} in Lemma 2 is obtained when the partial derivatives of $F_B \left( {\left\lbrace \z_{i} \right\rbrace}_{i \in \N}, {\left\lbrace y_{i} \right\rbrace}_{i \in \N} \right)$ \eqref{eq_fb_new} with respect to the elements of the variables $\{\z_{i}\}$ are set to zero. Similarly, the condition \eqref{eq_factor_global} in Lemma 4 are obtained when the partial derivatives of $F_B \left( {\left\lbrace \z_{i} \right\rbrace}_{i \in \N}, {\left\lbrace y_{i} \right\rbrace}_{i \in \N} \right)$ with respect to $\{y_i\}$ are set to zero. Hence, the properties derived in Lemmas 2 and 4 together constitute a sufficient condition for a stationary point of the BFE.

\subsection{Proof of Theorem \ref{thm_lf_cg}} \label{proof_lf_cg}
As discussed in Section \ref{subsec_acsma}, the global Gibbsian problem \eqref{opt_global} essentially solves a system of equations given in \eqref{eq_serv_fug}. From that analogy, it suffices to show that $\lf_i=[\lf_{ij}]_{j \in \N_i}$, the solution of the local Gibbsian optimization problem \eqref{opt_local_alternate0} at a link $i$,  is consistent with the following system of equations:
\begin{align}
s_j&= \sum\limits_{\x^{(i)} \in \mathcal{I}_i\; : \; x_j=1 } \frac{1}{Z_i} \Big(\prod\limits_{ k \in \N_i \; : \; x_k=1} e^{\lf_{ik}}\Big), \; \; \forall j \in \N_i, \label{eq_serv_cg} \\
\text{ where } &Z_i= \sum\limits_{\x^{(i)} \in \mathcal{I}_i } \Big(\prod\limits_{ k \in \N_i \; : \; x_k=1} e^{\lf_{ik}} \Big). \label{eq_z_cg}
\end{align}

For the conflict graph model, the above equations can be simplified as follows. In the conflict graph model, a link is active if and only if all its neighbours are inactive. Hence, there is only one local feasible schedule $\x^{(i)} \in \I_i$ in which link $i$ is active, namely $\x^{(i)}=(0,\dots,0,1,0,\dots,0)$ where $1$ is in the $i^{th}$ position. Using this observation in \eqref{eq_serv_cg} with $j=i$, we obtain
\begin{align}
s_i&= \frac{1}{Z_i}e^{\lf_{ii}}. \label{eq_si_cg}
\end{align}
Recall the definition of local feasibility from Section \ref{local_gibbs}. Any local schedule $\x^{(i)} \in \I_i$ at a link $i$, is feasible if that link $i$ is inactive. In other words, if $x_i=0$ in a local schedule $\x^{(i)} \in \{0,1\}^{\N_i}$, all the $2^{|\N_i|-1}$ combinations of its neighbours activations are allowed. This implies that for any $j \in \N_i \setminus \{i\}$, the set $\{\x^{(i)} \in \mathcal{I}_i\; : \; x_j=1 \}$ has all the $2^{|\N_i|-1}$ possible schedules. Using this observation in \eqref{eq_serv_cg} gives us
\begin{align}
s_j&= \frac{1}{Z_i} e^{\lf_{ij}} \prod\limits_{k \in \N_i \setminus \{i, j\}} (1+ e^{\lf_{ik}}), &j \in \N_i \setminus \{i\}. \label{eq_sj_cg}
\end{align}
Similarly, the normalization constant $Z_i$ \eqref{eq_z_cg} can be simplified to
\begin{align}
Z_i&= e^{\lf_{ii}} + \prod\limits_{j \in \N_i \setminus \{i\}} (1+ e^{\lf_{ij}}). \label{eq_z1_cg}
\end{align}
Note that \eqref{eq_si_cg}-\eqref{eq_z1_cg} characterize the relation between the local fugacities, and the service rates under the conflict graph model. Hence, it sufficient to prove that the local fugacities \eqref{eq_lf_cg} stated in Theorem \ref{thm_lf_cg}, are consistent with the set of equations \eqref{eq_si_cg}-\eqref{eq_z1_cg}. This step can be verified by simply substituting the local fugacity expressions \eqref{eq_lf_cg} in \eqref{eq_si_cg}-\eqref{eq_z1_cg}.
\subsection{Proof of Lemma \ref{lemma_sub_grad}} \label{proof_sub_grad}
The outline of the proof is to show that Algorithm 2 corresponds to the dual subgradient method for \eqref{eq_util_bethe}. First, we compute the Lagrangian, and the dual problem for \eqref{eq_util_bethe}. Considering the equality constraints in \eqref{eq_util_constraint}, the partial Lagrangian of \eqref{eq_util_bethe} is given by

\vspace{-3mm}
\begin{scriptsize}
\begin{align}
&L(y, \{\b_j\}; \be) \nonumber\\
&= \theta \sum_j U_j(y_j) + \sum_j H(\b_j) + \sum_{jk} \lf_{jk} \left(\sum_{\x^{(j)} : x_k=1} \b_j(\x^{(j)})-y_k\right). \label{eq_lagrangian}
\end{align}
\end{scriptsize}
\vspace{-3mm}

Here $\bm{\beta}:=\{\lf_{jk}\}$ is a short hand notation for the dual variables $\{\lf_{jk} , j =1 \dots N, k \in \nj\}$. These dual variables have an interpretation of local fugacities. Hence we abuse the notation by using the same notation for both of them.
The dual function is given by $D( \be) = \sup L(y, \{\b_j\}; \be)$ over  $y \in [0,1]^{\N}$ and the distributions $\{\b_{j}\}$ on the local schedules. We require the following result (Lemma \ref{lemma_libin_arxiv}) for completing the proof.

\begin{lemma} \label{lemma_libin_arxiv}
The primal and dual solutions of the optimization problem \eqref{eq_util_bethe} satisfy strong duality. Further, the primal and dual variables are related as follows. For a given dual value $\be=\{\lf_{jk}\}$, the Lagrangian $L(y, \{\b_j\}; \be)$ attains its supremum at primal values given by
\begin{align}
y_j(\bm{\lf})&= \arg \max_{q \in [0,1]} \theta U_j(q) - q \sum_{k \in \nj} \beta_{kj}(t), \; \forall j,\\
\b_j(\x^{(j)};\bm{\lf})&=  {Z_j^{-1}} \exp \Big(\sum_{k \in \N_j} x_k \beta_{jk}(t) \Big),  \forall \x^{(j)} \in \I_j, j \in \N.
\end{align}
\end{lemma}
\begin{proof}
The utility functions $\{U_i\}$ are concave, and the entropy is strictly concave. Hence, in \eqref{eq_util_bethe} we are maximizing a strictly concave function with affine constraints. Further, since the rate region $\Lambda_B$ is non empty, there always exist some $\{b_j\}$ and $y$ such that they are feasible for \eqref{eq_util_bethe}. Hence, the Slater's condition for convex problems with affine constraints \cite[Page 226]{boyd} implies strong duality.

For a given dual variable $\be=\{\lf_{jk}\}$, let $y_j(\bm{\lf})$, $\b_j(\x^{(j)};\bm{\lf})$ be the corresponding primal variables that maximize the Lagrangian. Then from the structure of the Lagrangian \eqref{eq_lagrangian}, it follows that 
\begin{align*}
y_j(\bm{\lf})&= \arg \max_{q \in [0,1]} \theta U_j(q) - q \sum_{k \in \nj} \beta_{kj}(t), \; \forall j.
\end{align*}
Next, consider the following partial derivative of the Lagrangian \eqref{eq_lagrangian} to obtain
\begin{align*}
\frac{\partial L(y, \{\b_j\}; \be) }{\partial \b_j(\x^{(j)})}= - \ln \b_j(\x^{(j)}) - 1 + \sum_{k \in \N_j: x_k=1} \beta_{jk}.
\end{align*}
Hence, the optimal value should satisfy
\begin{align*}
\b_j(\x^{(j)};\bm{\lf})& \;\alpha\; \exp \Big(\sum_{k \in \N_j} x_k \beta_{jk}(t) \Big), \forall \x^{(j)} \in \I_j, \forall j.
\end{align*}
\end{proof}
The subgradient for a given dual variable is equal to the residual error in the corresponding primal constraints (See \cite[Chapter 2]{book_subgradient} for details). Specifically, the subgradient at $\bm{\lf}$ denoted by $g(\bm{\lf}) := \{g_{jk} (\bm{\lf})\}$ is given by
\begin{align*}
g_{jk}(\be)=\Big(\sum_{\x^{(j)} : x_k=1} \b_j(\x^{(j)};\bm{\lf})\Big)-y_k(\bm{\lf}).
\end{align*}
Hence the update rule \eqref{eq_update} essentially corresponds to a dual subgradient method for \eqref{eq_util_bethe}. This completes the proof of Lemma \ref{lemma_sub_grad}.

\bibliographystyle{IEEEtran}
\bibliography{myreferences}

\begin{thebibliography}{10}
\providecommand{\url}[1]{#1}
\csname url@samestyle\endcsname
\providecommand{\newblock}{\relax}
\providecommand{\bibinfo}[2]{#2}
\providecommand{\BIBentrySTDinterwordspacing}{\spaceskip=0pt\relax}
\providecommand{\BIBentryALTinterwordstretchfactor}{4}
\providecommand{\BIBentryALTinterwordspacing}{\spaceskip=\fontdimen2\font plus
\BIBentryALTinterwordstretchfactor\fontdimen3\font minus
  \fontdimen4\font\relax}
\providecommand{\BIBforeignlanguage}[2]{{%
\expandafter\ifx\csname l@#1\endcsname\relax
\typeout{** WARNING: IEEEtran.bst: No hyphenation pattern has been}%
\typeout{** loaded for the language `#1'. Using the pattern for}%
\typeout{** the default language instead.}%
\else
\language=\csname l@#1\endcsname
\fi
#2}}
\providecommand{\BIBdecl}{\relax}
\BIBdecl

\bibitem{allerton_bethe}
P.~S. Swamy, R.~K. Ganti, and K.~Jagannathan, ``Adaptive {CSMA} under the
  {SINR} model: Fast convergence through local {Gibbs} otpimization,'' in
  \emph{Communication, Control, and Computing (Allerton), 2015 53rd Annual
  Allerton Conference on}.\hskip 1em plus 0.5em minus 0.4em\relax IEEE, 2015.

\bibitem{tassiulas}
L.~Tassiulas and A.~Ephremides, ``Stability properties of constrained queueing
  systems and scheduling policies for maximum throughput in multihop radio
  networks,'' \emph{Automatic Control, IEEE Transactions on}, vol.~37, no.~12,
  pp. 1936--1948, 1992.

\bibitem{TE93}
------, ``{Dynamic server allocation to parallel queues with randomly varying
  connectivity},'' \emph{IEEE Transactions on Information Theory}, vol.~39,
  no.~2, pp. 466--478, 1993.

\bibitem{chaporkar2008throughput}
P.~Chaporkar, K.~Kar, X.~Luo, and S.~Sarkar, ``{Throughput and fairness
  guarantees through maximal scheduling in wireless networks},'' \emph{IEEE
  Transactions on Information Theory}, vol.~54, no.~2, p. 572, 2008.

\bibitem{dimakis2006sufficient}
A.~Dimakis and J.~Walrand, ``Sufficient conditions for stability of
  longest-queue-first scheduling: Second-order properties using fluid limits,''
  \emph{Advances in Applied Probability}, pp. 505--521, 2006.

\bibitem{wu06}
X.~Wu, R.~Srikant, and J.~Perkins, ``{Queue-length stability of maximal greedy
  schedules in wireless networks},'' in \emph{Proceedings of Information Theory
  and Applications Inaugural Workshop}, 2006, pp. 6--10.

\bibitem{libin}
L.~Jiang and J.~Walrand, ``A distributed {CSMA} algorithm for throughput and
  utility maximization in wireless networks,'' \emph{IEEE/ACM Transactions on
  Networking (TON)}, vol.~18, no.~3, pp. 960--972, 2010.

\bibitem{dshah}
S.~Rajagopalan, D.~Shah, and J.~Shin, ``Network adiabatic theorem: an efficient
  randomized protocol for contention resolution,'' in \emph{ACM SIGMETRICS
  Performance Evaluation Review}, vol.~37, no.~1.\hskip 1em plus 0.5em minus
  0.4em\relax ACM, 2009, pp. 133--144.

\bibitem{qcsma}
J.~Ni, B.~Tan, and R.~Srikant, ``Q-{CSMA}: Queue-length-based {CSMA/CA}
  algorithms for achieving maximum throughput and low delay in wireless
  networks,'' \emph{Networking, IEEE/ACM Transactions on}, vol.~20, no.~3, pp.
  825--836, 2012.

\bibitem{bremaud}
P.~Bremaud, \emph{Markov chains: Gibbs fields, {Monte Carlo} simulation, and
  queues}.\hskip 1em plus 0.5em minus 0.4em\relax springer, 1999, vol.~31.

\bibitem{hardness}
D.~Shah, D.~N. Tse, and J.~N. Tsitsiklis, ``Hardness of low delay network
  scheduling,'' \emph{Information Theory, IEEE Transactions on}, vol.~57,
  no.~12, pp. 7810--7817, 2011.

\bibitem{fast_mixing}
L.~Jiang, M.~Leconte, J.~Ni, R.~Srikant, and J.~Walrand, ``Fast mixing of
  parallel {G}lauber dynamics and low-delay {CSMA} scheduling,''
  \emph{Information Theory, IEEE Transactions on}, vol.~58, no.~10, pp.
  6541--6555, 2012.

\bibitem{parallel_chains}
D.~Lee, D.~Yun, J.~Shin, Y.~Yi, and S.-Y. Yun, ``Provable per-link
  delay-optimal {CSMA} for general wireless network topology,'' in
  \emph{INFOCOM, 2014 Proceedings IEEE}.\hskip 1em plus 0.5em minus 0.4em\relax
  IEEE, 2014, pp. 2535--2543.

\bibitem{graph_limitations}
M.~M. Halldorsson and T.~Tonoyan, ``How well can graphs represent wireless
  interference?'' \emph{arXiv preprint arXiv:1411.1263}, 2014.

\bibitem{sinr_mimo}
D.~Qian, D.~Zheng, J.~Zhang, and N.~Shroff, ``{CSMA}-based distributed
  scheduling in multi-hop {MIMO} networks under {SINR} model,'' in
  \emph{INFOCOM, 2010 Proceedings IEEE}.\hskip 1em plus 0.5em minus 0.4em\relax
  IEEE, 2010, pp. 1--9.

\bibitem{sinr_mimo_journal}
J.-G. Choi, C.~Joo, J.~Zhang, and N.~B. Shroff, ``Distributed link scheduling
  under {SINR} model in multihop wireless networks,'' \emph{IEEE/ACM
  Transactions on Networking (TON)}, vol.~22, no.~4, pp. 1204--1217, 2014.

\bibitem{ncc_paper}
P.~S. Swamy, R.~K. Ganti, and K.~Jagannathan, ``Spatial {CSMA}: A distributed
  scheduling algorithm for the {SIR} model with time-varying channels,'' in
  \emph{Communications (NCC), 2015 Twenty First National Conference on}.\hskip
  1em plus 0.5em minus 0.4em\relax IEEE, 2015, pp. 1--6.

\bibitem{bp_csma}
C.~H. Kai and S.~C. Liew, ``Applications of belief propagation in {CSMA}
  wireless networks,'' \emph{Networking, IEEE/ACM Transactions on}, vol.~20,
  no.~4, pp. 1276--1289, 2012.

\bibitem{bethe_jshin}
S.-Y. Yun, J.~Shin, and Y.~Yi, ``{CSMA} using the {Bethe} approximation for
  utility maximization,'' in \emph{Information Theory Proceedings (ISIT), 2013
  IEEE International Symposium on}.\hskip 1em plus 0.5em minus 0.4em\relax
  IEEE, 2013, pp. 206--210.

\bibitem{yedidia}
J.~S. Yedidia, W.~T. Freeman, and Y.~Weiss, ``Constructing free-energy
  approximations and generalized belief propagation algorithms,''
  \emph{Information Theory, IEEE Transactions on}, vol.~51, no.~7, pp.
  2282--2312, 2005.

\bibitem{bethe_emprical}
K.~P. Murphy, Y.~Weiss, and M.~I. Jordan, ``Loopy belief propagation for
  approximate inference: An empirical study,'' in \emph{Proceedings of the
  Fifteenth conference on Uncertainty in artificial intelligence}.\hskip 1em
  plus 0.5em minus 0.4em\relax Morgan Kaufmann Publishers Inc., 1999, pp.
  467--475.

\bibitem{parallel_chains2}
J.~Kwak, C.-H. Lee, and D.~Y. Eun, ``Exploiting the past to reduce delay in
  csma scheduling: a high-order markov chain approach,'' in \emph{ACM
  SIGMETRICS Performance Evaluation Review}, vol.~41, no.~1.\hskip 1em plus
  0.5em minus 0.4em\relax ACM, 2013, pp. 353--354.

\bibitem{baccelli}
F.~Baccelli and C.~Singh, ``Adaptive spatial {Aloha}, fairness and stochastic
  geometry,'' in \emph{Modeling \& Optimization in Mobile, Ad Hoc \& Wireless
  Networks (WiOpt), 2013 11th International Symposium on}.\hskip 1em plus 0.5em
  minus 0.4em\relax IEEE, 2013, pp. 7--14.

\bibitem{radius_approx1}
G.~Brar, D.~M. Blough, and P.~Santi, ``Computationally efficient scheduling
  with the physical interference model for throughput improvement in wireless
  mesh networks,'' in \emph{Proceedings of the 12th annual international
  conference on Mobile computing and networking}.\hskip 1em plus 0.5em minus
  0.4em\relax ACM, 2006, pp. 2--13.

\bibitem{radius_approx2}
P.~C. Pinto and M.~Z. Win, ``Communication in a {P}oisson field of
  interferers,'' in \emph{Information Sciences and Systems, 2006 40th Annual
  Conference on}.\hskip 1em plus 0.5em minus 0.4em\relax IEEE, 2006, pp.
  432--437.

\bibitem{libin_book}
L.~Jiang and J.~Walrand, ``Scheduling and congestion control for wireless and
  processing networks,'' \emph{Synthesis Lectures on Communication Networks},
  vol.~3, no.~1, pp. 1--156, 2010.

\bibitem{boyd}
S.~Boyd and L.~Vandenberghe, \emph{Convex optimization}.\hskip 1em plus 0.5em
  minus 0.4em\relax Cambridge university press, 2009.

\bibitem{book_martin}
M.~J. Wainwright and M.~I. Jordan, ``Graphical models, exponential families,
  and variational inference,'' \emph{Foundations and Trends{\textregistered} in
  Machine Learning}, vol.~1, no. 1-2, pp. 1--305, 2008.

\bibitem{book_subgradient}
N.~Z. Shor, \emph{Minimization methods for non-differentiable functions}.\hskip
  1em plus 0.5em minus 0.4em\relax Springer Science \& Business Media, 2012,
  vol.~3.

\bibitem{libin_arxiv}
L.~Jiang, D.~Shah, J.~Shin, and J.~Walrand, ``Distributed random access
  algorithm: Scheduling and congesion control,'' \emph{arXiv preprint
  arXiv:0907.1266}, 2009.

\end{thebibliography}

\end{document}